\documentclass[journal]{IEEEtran}

\usepackage{graphics,epsf,psfrag,amsmath,amssymb,amsfonts,verbatim,cite}
\usepackage[mathcal]{eucal}
\usepackage{mathrsfs}

\def\ACSP{/research/ACSP}

\newtheorem{theorem}{Theorem}

\newtheorem{lemma}{Lemma}

\newcommand{\Hmsc}{\mathscr{H}}

\newcommand{\Vmsc}{\mathscr{V}}

\newcommand{\Xmsc}{\mathscr{X}}

\newcommand{\Acal}{\mathcal{A}}
\newcommand{\Bcal}{\mathcal{B}}
\newcommand{\Ccal}{\mathcal{C}}

\newcommand{\Ecal}{\mathcal{E}}

\newcommand{\Hcal}{\mathcal{H}}

\newcommand{\Mcal}{\mathcal{M}}

\newcommand{\Pcal}{\mathcal{P}}
\newcommand{\Qcal}{\mathcal{Q}}
\newcommand{\Rcal}{\mathcal{R}}
\newcommand{\Scal}{\mathcal{S}}
\newcommand{\Tcal}{\mathcal{T}}
\newcommand{\Ucal}{\mathcal{U}}

\newcommand{\Wcal}{\mathcal{W}}
\newcommand{\Xcal}{\mathcal{X}}
\newcommand{\Ycal}{\mathcal{Y}}
\newcommand{\Zcal}{\mathcal{Z}}

\newcommand{\Rmbb}{\mathbb{R}}
\newcommand{\Embb}{\mathbb{E}}

\newcommand{\beq}{\begin{equation}}
\newcommand{\eeq}{\end{equation}}
\newcommand{\eps}{\epsilon}

\title{Distributed Source Coding in the Presence of Byzantine Sensors}

\author{Oliver~Kosut,~\IEEEmembership{Student Member, IEEE} and~Lang~Tong,~\IEEEmembership{Fellow,~IEEE}
\thanks{This work is supported in part  by the National
Science Foundation under award CCF-0635070,
the U. S. Army Research Laboratory under
the Collaborative Technology Alliance Program
DAAD19-01-2-0011, and TRUST  (The Team for Research in
Ubiquitous Secure Technology) sponsored by
the National Science Foundation under award
CCF-0424422.}}

\begin{document}
\maketitle

\begin{abstract}
The distributed source coding problem is considered when the sensors, or encoders, are under Byzantine attack; that is, an unknown group of sensors have been reprogrammed by a malicious intruder to undermine the reconstruction at the fusion center. Three different forms of the problem are considered. The first is a variable-rate setup, in which the decoder adaptively chooses the rates at which the sensors transmit. An explicit characterization of the variable-rate achievable sum rates is given for any number of sensors and any groups of traitors. The converse is proved constructively by letting the traitors simulate a fake distribution and report the generated values as the true ones. This fake distribution is chosen so that the decoder cannot determine which sensors are traitors while maximizing the required rate to decode every value. Achievability is proved using a scheme in which the decoder receives small packets of information from a sensor until its message can be decoded, before moving on to the next sensor. The sensors use randomization to choose from a set of coding functions, which makes it probabilistically impossible for the traitors to cause the decoder to make an error. Two forms of the fixed-rate problem are considered, one with deterministic coding and one with randomized coding. The achievable rate regions are given for both these problems, and it is shown that lower rates can be achieved with randomized coding.
\end{abstract}


\section{Introduction}

\PARstart{W}{e} consider a modification to the distributed source coding problem in which an unknown subset of sensors are taken over by a malicious intruder and reprogrammed. We assume there are $m$ sensors. Each time slot, sensors $i$ for $i=1,\cdots,m$ observe random variables $X_i$ according to the joint probability distribution $p(x_1\cdots x_m)$. Each sensor encodes its observation independently and transmits a message to a common decoder, which attempts to reconstruct the source values with small probability of error based on those messages. A subset of sensors are \emph{traitors}, while the rest are \emph{honest}. Unbeknownst to the honest sensors or the decoder, the traitors have been reprogrammed to cooperate to obstruct the goal of the network, launching a so-called Byzantine attack. To counter this attack, the honest sensors and decoder must employ strategies so that the decoder can correctly reconstruct source values no matter what the traitors do.

It is obvious that observations made by the traitors are irretrievable unless the traitors choose to deliver them to the decoder.  Thus the best the decoder can hope to achieve is to reconstruct the observations of the honest sensors. A simple procedure is to ignore the statistical correlations among the observations and collect data from each sensor individually.  The total sum rate of such an approach is  $\sum_i H(X_i)$.   One expects however that this sum rate can be lowered if the correlation structure is not ignored.

Without traitors, Slepian-Wolf coding \cite{Slepian&Wolf:73IT} can be used to achieve a sum rate as low as
\beq H(X_1\cdots X_m).\label{eq:swsumrate}\eeq
However, standard Slepian-Wolf coding has no mechanism for handling any deviations from the agreed-upon encoding functions by the sensors. Even a random fault by a single sensor could have devastating consequences for the accuracy of the source estimates produced at the decoder, to say nothing of a Byzantine attack on multiple sensors. In particular, because Slepian-Wolf coding takes advantage of the correlation among sources, manipulating the codeword for one source can alter the accuracy of the decoder's estimate for other sources. It will turn out that for most source distributions, the sum rate given in \eqref{eq:swsumrate} cannot be achieved if there is even a single traitor.

In this paper, we are interested in  the lowest achievable sum-rate
such that the decoder can reconstruct observations of the honest sensors with arbitrarily
small error probability.  In some cases, we are also interested in the 
rate region.   We note that although the problem setup does not allow the detector 
to distinguish traitors from the honest sensors, an efficient scheme
that guarantees the reconstruction of data from honest sensors
is of both theoretical and practical interest.  For example, 
for a distributed inference problem in the presence of 
Byzantine sensors, a practical (though not necessarily optimal) solution is to attack the problem in two 
separate phases . In the first phase,
the decoder collects data from sensors over multiple access channels
with rate constraints. Here we require that data from honest sensors
are perfectly reconstructed at the decoder even though the decoder
does not know which piece of data is from an honest sensor.  
In the second step, the received data is used for statistical inference.
The example of distributed detection in the presence of Byzantine
sensors is considered in \cite{Marano&Matta&Tong:06Asilomar}.
The decoder may also have other side information about the content of
the messages that allows the decoder to distinguish messages from
the honest sensors.

\subsection{Related Work}

The notion of Byzantine attack has its root in the Byzantine generals problem \cite{Lamport&Shostak&Pease:82ACM,Dolev:82} in which
a clique of traitorous generals conspire to prevent loyal generals from forming consensus. It was shown in \cite{Lamport&Shostak&Pease:82ACM} that consensus in the presence of Byzantine attack is possible if and only if less than a third of the generals are traitors.

Countering Byzantine attacks in communication networks has also been studied in the past by many authors.  See the earlier work of Perlman \cite{Perlman:88thesis} and also more recent review \cite{Zhou&Haas:99,Hu&Perrig:04}. An information theoretic network coding approach to Byzantine attack is presented in \cite{Ho&etal:04ISIT}. In \cite{Awerbuch&etal:02WISE}, Awerbuch et al suggest a method for mitigating Byzantine attacks on routing in ad hoc networks. Their approach is most similar to ours in the way they maintain a list of current knowledge about which links are trustworthy, constantly updated based on new information. Sensor fusion with Byzantine sensors was studied in \cite{Kosut&Tong:06Allerton}. In that paper, the sensors, having already agreed upon a message, communicate it to the fusion center over a discrete memoryless channel. Quite similar results were shown in \cite{Jaggi&etal:05ISIT}, in which a malicious intruder takes control of a set of links in the network. The authors show that two nodes can communicate at a nonzero rate as long as less than half of the links between them are Byzantine. This is different from the current paper in that the transmitter chooses its messages, instead of relaying information received from an outside source, but some of the same approaches from \cite{Jaggi&etal:05ISIT} are used in the current paper, particularly the use of randomization to fool traitors that have already transmitted.

\subsection{Redefining Achievable Rate}\label{subsection:redefinition}

The nature of Byzantine attack require three modifications to the usual notion of achievable rate. The first, as mentioned above, is that small probability of error is required only for honest sources, even though the decoder may not know which sources are honest. This requirement is reminiscent of \cite{Lamport&Shostak&Pease:82ACM}, in which the lieutenants need only perform the commander's order if the commander is not a traitor, even though the lieutenants might not be able to decide this with certainty.

The next modification is that there must be small probability of error no matter what the traitors do. This is essentially the definition of Byzantine attack.

The final modification has to do with which sensors are allowed to be traitors. Let $\Hcal$ be the set of honest sensors, and $\Tcal=\{1,\cdots,m\}\backslash\Hcal$ the set of traitors. Any code is associated with a list of which sets of sensors it can handle as the set of traitors. A rate is then achieved if the code gets small probability of error when the actual set of traitors is in fact on the list. It will be more convenient to specify not the list of allowable sets of traitors, but rather the list of allowable sets of honest sensors. We define $\Hmsc\subset 2^{\{1,\cdots,m\}}$ to be this list. Thus small probability of error is required only when $\Hcal\in\Hmsc$. One special case is when the code can handle any group of at most $t$ traitors. That is,
\[\Hmsc=\Hmsc_t\triangleq\{\Scal\subset\{1,\cdots,m\}:|\Scal|\ge m-t\}.\]
Observe that achievable rates depend not just on the true set of traitors but also on the collection $\Hmsc$, because the decoder's willingness to accept more and more different groups of traitors allows the true traitors to get away with more without being detected. Thus we see a trade off between rate and security---in order to handle more traitors, one needs to be willing to accept a higher rate.

\subsection{Fixed-Rate Versus Variable-Rate Coding}

In standard source coding, an encoder is made up of a single encoding function. We will show that this fixed-rate setup is suboptimal for this problem, in the sense that we can achieve lower sum rates using variable-rate coding. By variable-rate we mean that the number of bits transmitted per source value by a particular sensor will not be fixed. Instead, the decoder chooses the rates at ``run time'' in the following way. Each sensor has a finite number of encoding functions, all of them fixed beforehand, but with potentially different output alphabets. The coding session is then made up of a number of transactions. Each transaction begins with the decoder deciding which sensor will transmit, and which of its several encoding functions it will use. The sensor then executes the chosen encoding function and transmits the output back to the decoder. Finally, the decoder uses the received message to choose the next sensor and encoding function, beginning the next transaction, and so on. Thus a code is made up of a set of encoding functions for each sensor, a method for the decoder to choose sensors and encoding functions based on previously received messages, and lastly a decoding function that takes all received messages and produces source estimates.

Note that the decoder has the ability to transmit some information back to the sensors, but this feedback is limited to the choice of encoding function. Since the number of encoding functions need not grow with the block length, this represents zero rate feedback.

In variable-rate coding, since the rates are only decided upon during the coding session, there is no notion of an $m$-dimensional achievable rate region. Instead, we only discuss achievable sum rates.

\subsection{Traitor Capabilities}

An important consideration with Byzantine attack is the information to which the traitors have access. First, we assume that the traitors have complete knowledge of the coding scheme used by the decoder and honest sensors. Furthermore, we always assume that they can communicate with each other arbitrarily. For variable-rate coding, they may have any amount of ability to eavesdrop on transmissions between honest sensors and the decoder. We will show that this ability has no effect on achievable rates. We assume with fixed-rate coding that all sensors transmit simultaneously, so it does not make sense that traitors could eavesdrop on honest sensors' transmissions before making their own, as that would violate causality. Thus we assume for fixed-rate coding that the traitors cannot eavesdrop.

The key factor, however, is the extent to which the traitors have direct access to information about the sources. We assume the most general memoryless case, that the traitors have access to the random variable $W$, where $W$ is i.i.d. distributed with $(X_1\cdots X_m)$ according to the conditional distribution $r(w|x_1\cdots x_m)$. A natural assumption would be that $W$ always includes $X_i$ for traitors $i$, but in fact this need not be the case. An important special case is where $W=(X_1,\cdots,X_m)$, i.e. the traitors have perfect information.

We assume that the distribution of $W$ depends on who the traitors are, and that the decoder may not know exactly what this distribution is. Thus each code is associated with a function $\Rcal$ that maps elements of $\Hmsc$ to sets of conditional distributions $r$. The relationship between $r$ and $\Rcal(\Hcal)$ is analogous to the relationship between $\Hcal$ and $\Hmsc$. That is, given $\Hcal$, the code is willing to accept all distributions $r\in\Rcal(\Hcal)$. Therefore a code is designed based on $\Hmsc$ and $\Rcal$, and then the achieved rate depends at run time on $\Hcal$ and $r$, where we assume $\Hcal\in\Hmsc$ and $r\in\Rcal(\Hcal)$. We therefore discuss not achievable rates $R$ but rather achievable rate functions $R(\Hcal,r)$. In fact, this applies only to variable-rate codes. In the fixed-rate case, no run time rate decisions can be made, so achievable rates depend only on $\Hmsc$ and $\Rcal$.

\subsection{Main Results}

The main results of this paper give explicit characterizations of the achievable rates for three different setups. The first, which is discussed in the most depth, is the variable-rate case, for which we characterize achievable sum rate functions. The other two setups are for fixed-rate coding, divided into deterministic and randomized coding, for which we give $m$-dimensional achievable rate regions. We show that randomized coding yields a larger achievable rate region than deterministic coding, but we believe that in most cases randomized fixed-rate coding requires an unrealistic assumption. In addition, even randomized fixed-rate coding cannot achieve the same sum rates as variable-rate coding.

We give the exact solutions in Theorems 1 and 2, but describe here the intuition behind them. For variable-rate, the achievable rates are based on alternate distributions on $(X_1\cdots X_m)$. Specifically, given $W$, the traitors can simulate any distribution $\bar{q}(x_\Tcal|w)$ to produce a fraudulent version of $X_\Tcal^n$, then report this sequence as the truth. Suppose that the overall distribution $q(x_1\cdots x_m)$ governing the combination of the true value of $X_\Hcal^n$ with this fake value of $X_\Tcal^n$ could be produced in several different ways, with several different sets of traitors. In that case, the decoder cannot tell which of these several possibilities is the truth, which means that from its point of view, any sensor that is honest in one of these possibilities may in fact be honest. Since the error requirement described in \ref{subsection:redefinition} stipulates that the decoder must produce a correct estimate for every honest sensor, it must attempt to decode the source values associated with all these potentially honest sensors. Thus the sum rate must be at least the joint entropy, when distributed according to $q$, of the sources associated with all potentially honest sensors. The supremum over all such $\bar{q}$s is the achievable sum rate.

For example, suppose $\Hmsc=\Hmsc_{m-1}$. That is, at most one sensor is honest. Then the traitors are able to create the distribution $q(x_1\cdots x_m)=p(x_1)\cdots p(x_m)$ no matter what group of $m-1$ sensors are the traitors. Thus every sensor appears as if it could be the honest one, so the minimum achievable sum rate is
\beq H(X_1)+\cdots +H(X_m).\label{eq:tm1}\eeq
In other words, the decoder must use an independent source code for each sensor, which requires receiving $nH(X_i)$ bits from sensor $i$ for all $i$.

The achievable fixed-rate regions are based on the Slepian-Wolf achievable rate region. For randomized fixed-rate coding, the achievable region is such that for all $\Scal\in\Hmsc$, the rates associated with the sensors in $\Scal$ fall into the Slepian-Wolf rate region on the corresponding random variables. Note that for $\Hmsc=\{\{1,\cdots,m\}\}$, this is identical to the Slepian-Wolf region. For $\Hmsc=\Hmsc_{m-1}$, this region is such that for all $i$, $R_i\ge H(X_i)$, which corresponds to the sum rate in \eqref{eq:tm1}. The deterministic fixed-rate achievable region is a subset of that of randomized fixed-rate, but with an additional constraint stated in Section~\ref{section:fixedrate}.

\subsection{Randomization}\label{subsection:random}

Randomization plays a key role in defeating Byzantine attacks. As we have discussed, allowing randomized encoding in the fixed-rate situation expands the achievable region. In addition, the variable-rate coding scheme that we propose relies heavily on randomization to achieve small probability of error. In both fixed and variable-rate coding, randomization is used as follows. Every time a sensor transmits, it randomly chooses from a group of essentially identical encoding functions. The index of the chosen function is transmitted to the decoder along with its output. Without this randomization, a traitor that transmits before an honest sensor $i$ would know exactly the messages that sensor $i$ will send. In particular, it would be able to find fake sequences for sensor $i$ that would produce those same messages. If the traitor tailors the messages it sends to the decoder to match one of those fake sequences, when sensor $i$ then transmits, it would appear to corroborate this fake sequence, causing an error. By randomizing the choice of encoding function, the set of sequences producing the same message is not fixed, so a traitor can no longer know with certainty that a particular fake source sequence will result in the same messages by sensor $i$ as the true one. This is not unlike Wyner's wiretap channel \cite{Wyner:75BSTJ}, in which information is kept from the wiretapper by introducing additional randomness. See in particular Section~\ref{subsection:errorprob} for the proof that variable-rate randomness can defeat the traitors in this manner.


The rest of the paper is organized as follows. 
In Section~\ref{section:example}, we develop in detail the case that there are three sensors and one traitor, describing a coding scheme that achieves the optimum sum rate. 
In Section~\ref{section:model}, we formally give the variable-rate model and present the variable-rate result. 
In Section~\ref{section:comments}, we discuss the variable-rate achievable rate region and give an analytic formulation for the minimum achievable sum rate for some special cases. 
In Section~\ref{section:fixedrate}, we give the fixed-rate models and present the fixed-rate result. In Sections~\ref{section:proof} and~\ref{section:prooffixed}, we prove the variable-rate and fixed-rate results respectively. Finally, in Section~\ref{section:conclusion}, we conclude.

\section{Three Sensor Example}\label{section:example}

\subsection{Potential Traitor Techniques}

For simplicity and motivation, we first explore the three-sensor case with one traitor. That is, $m=3$ and \[\Hmsc=\{\{1,2\},\{2,3\},\{1,3\}\}.\]
Suppose also that the traitor has access to perfect information. Consider first the simple case where the $X_i$ can be decomposed as
\begin{align*}
X_1&=(Y_1,Y_{12},Y_{13},Y_{123}),\\
X_2&=(Y_2,Y_{12},Y_{23},Y_{123}),\\
X_3&=(Y_3,Y_{13},Y_{23},Y_{123})
\end{align*}
where $Y_1,Y_2,Y_3,Y_{12},Y_{13},Y_{23},Y_{123}$ are independent. Suppose the traitor is sensor 3. It can generate a new, independent version of $Y_{23}$, call it $Y'_{23}$, and then form $X'_3=(Y_1,Y_{13},Y'_{23},Y_{123})$. We claim that if sensor 3 now behaves for the rest of the coding session as if this counterfeit $X'_3$ were the real value, then the decoder will not be able to determine the traitor's identity. This is because both $(X_1,X_2)$ and $(X_2,X'_3)$ look like they could be a true pair, since all information that they share matches. Thus the decoder cannot know which of sensors 1 or 3 is the traitor, and which of $Y_{23}$ or $Y'_{23}$ is the truth, so it must obtain estimates of them both. To construct estimates of all three variables, every piece except $Y_{23}$ must be received only once, but the two versions $Y_{23}$ must be received separately. Therefore the sum rate must be at least
\beq H(X_1X_2X_3)+H(Y_{23})=H(X_1X_2X_3)+I(X_2;X_3|X_1).\label{eq:simpleexample}\eeq
In fact, this last expression holds for general distributions as well, as we demonstrate next.

Now take any distribution $p$, again with sensor 3 as the traitor. Sensors 1 and 2 will behave honestly, so they will report $X_1$ and $X_2$ correctly, as distributed according to the marginal distribution $p(x_1x_2)$. Since sensor 3 has access to the exact values of $X_1$ and $X_2$, it may simulate the conditional distribution $p(x_3|x_2)$, then take the resulting $X_3$ sequence and report it as the truth. Effectively, then, the three random variables will be distributed according to the distribution 
\[q(x_1x_2x_3)\triangleq p(x_1x_2)p(x_3|x_2).\] 
The decoder will be able to determine that sensors 1 and 2 are reporting jointly typical sequences, as are sensors 2 and 3, but not sensors 1 and 3. Therefore, it can tell that either sensor 1 or 3 is the traitor, but not which one, so it must obtain estimates of the sources from all three sensors. Since the three streams are not jointly typical with respect to the source distribution $p(x_1x_2x_3)$, standard Slepian-Wolf coding on three encoders will not correctly decode them all. However, had we known the strategy of the traitor, we could do Slepian-Wolf coding with respect to the distribution $q$. This will take a sum rate of
\[H_q(X_1X_2X_3)=H(X_1X_2X_3)+I(X_1;X_3|X_2)\]
where $H_q$ is the entropy with respect to $q$. In fact we will not do Slepian-Wolf coding with respect to $q$ but rather something slightly different that gives the same rate. Observe that this matches \eqref{eq:simpleexample}. Since Slepian-Wolf coding without traitors can achieve a sum rate of $H(X_1X_2X_3)$, we have paid a penalty of $I(X_1;X_3|X_2)$ for the single traitor. 

We supposed that sensor 3 simulated the distribution $p(x_3|x_2)$. It could have just as easily simulated $p(x_3|x_1)$, or another sensor could have been the traitor. Hence, the minimum achievable sum rate  for all $\Hcal\in\Hmsc$ is at least
\begin{multline}R^{*}\triangleq H(X_1X_2X_3)+\max\{I(X_1;X_2|X_3),\\I(X_1;X_3|X_2),I(X_2;X_3|X_1)\}.\label{eq:3rate}\end{multline}
In fact, this is exactly the minimum achievable sum rate, as shown below.

\subsection{Variable-Rate Coding Scheme}

We now give a variable-rate coding scheme that achieves $R^*$. This scheme is somewhat different from the one we present for the general case in Section~\ref{section:proof}, but it is much simpler, and it illustrates the basic idea. The procedure will be made up of a number of rounds. Communication from sensor $i$ in the first round will be based solely on the first $n$ values of $X_i$, in the second round on the second $n$ values of $X_i$, and so on. The principle advantage of the round structure is that the decoder may hold onto information that is carried over from one round to the next.

In particular, the decoder maintains a collection $\Vmsc\subset\Hmsc$ representing the sets that could be the set of honest sensors. If a sensor is completely eliminated from $\Vmsc$, that means it has been identified as the traitor. We begin with $\Vmsc=\Hmsc$, and then remove a set from $\Vmsc$ whenever we find that the messages from the corresponding pair of sensors are not jointly typical. With high probability, the two honest sensors report jointly typical sequences, so we expect never to eliminate the honest pair from $\Vmsc$. If the traitor employs the $q$ discussed above, for example, we would expect sensors 1 and 3 to report atypical sequences, so we will drop $\{1,3\}$ from $\Vmsc$. In essence, the value of $\Vmsc$ contains our current knowledge about what the traitor is doing.

The procedure for a round is as follows. If $\Vmsc$ contains $\{\{1,2\},\{1,3\}\}$, do the following:
\begin{enumerate}
\item Receive $nH(X_1)$ bits from sensor 1 and decode $x_1^n$.
\item Receive $nH(X_2|X_1)$ bits from sensor 2. If there is a sequence in $\Xcal_2^n$ jointly typical with $x_1^n$ that matches this transmission, decode that sequence to $x_2^n$. If not, receive $nI(X_1;X_2)$ additional bits from sensor 2, decode $x_2^n$, and remove $\{1,2\}$ from $\Vmsc$.
\item Do the same with sensor 3: Receive $nH(X_3|X_1)$ bits and decode $x_3^n$ if possible. If not, receive $nI(X_1;X_3)$ additional bits, decode, and remove $\{1,3\}$ from $\Vmsc$.
\end{enumerate}
If $\Vmsc$ is one of the other two subsets of $\Hmsc$ with two elements, perform the same procedure but replace sensor 1 with whichever sensor appears in both elements in $\Vmsc$. If $\Vmsc$ contains just one element, then we have exactly identified the traitor, so ignore the sensor that does not appear and simply do Slepian-Wolf coding on the two remaining sensors.

Note that the only cases when the number of bits transmitted exceeds $nR^*$ are when we receive a second message from one of the sensors, which happens exactly when we eliminate an element from $\Vmsc$. Assuming the source sequences of the two honest sensors are jointly typical, this can occur at most twice, so we can always achieve a sum rate of $R^*$ when averaged over enough rounds.

\subsection{Fixed-Rate Coding Scheme}

In the procedure described above, the number of bits sent by a sensor changes from round to round. We can no longer do this with fixed-rate coding, so we need a different approach. Suppose sensor 3 is the traitor. It could perform a black hole attack, in which case the estimates for $X_1^n$ and $X_2^n$ must be based only on the messages from sensors 1 and 2. Thus, the rates $R_1$ and $R_2$ must fall into the Slepian-Wolf achievability region for $X_1$ and $X_2$. Similarly, if one of the other sensors was the traitor, the other pairs of rates also must fall into the corresponding Slepian-Wolf region. Putting these conditions together gives
\beq
\begin{gathered}
\begin{aligned}
R_1&\ge \max\{H(X_1|X_2),H(X_1|X_3)\}\\
R_2&\ge \max\{H(X_2|X_1),H(X_2|X_3)\}\\
R_3&\ge \max\{H(X_3|X_1),H(X_3|X_2)\}
\end{aligned}\\
\begin{aligned}
R_1+R_2&\ge H(X_1X_2)\\
R_1+R_3&\ge H(X_1X_3)\\
R_2+R_3&\ge H(X_2X_3).
\end{aligned}
\end{gathered}\label{eq:3frregion}
\eeq
If the rates fall into this region, we can do three simultaneous Slepian-Wolf codes, one on each pair of sensors, thereby constructing two estimates for each sensor. If we randomize these codes using the method described in Section~\ref{subsection:random}, the traitor will be forced either to report the true message, or report a false message, which with high probability will be detected as such. Thus either the two estimates for each sensor will be the same, in which case we know both are correct, or one of the estimates will be demonstrably false, in which case the other is correct.

We now show that the region given by \eqref{eq:3frregion} does not include sum rates as low as $R^*$. Assume without loss of generality that $I(X_1;X_2|X_3)$ achieves the maximum in \eqref{eq:3rate}. Summing the last three conditions in \eqref{eq:3frregion} gives
\begin{multline}
R_1+R_2+R_3\ge \frac{1}{2}\big(H(X_1X_2)+H(X_1X_3)+H(X_2X_3)\big)\\
=H(X_1X_2X_3)+\frac{1}{2}\big(I(X_1;X_2|X_3)+I(X_1X_2;X_3)\big).
\label{eq:fixed3rate}
\end{multline}
If $I(X_1X_2;X_3)>I(X_1;X_2|X_3)$, \eqref{eq:fixed3rate} is larger than \eqref{eq:3rate}. Hence, there exist source distributions for which we cannot achieve the same sum rates with even randomized fixed-rate coding as with variable-rate coding.

If we are interested only in deterministic codes, the region given by \eqref{eq:3frregion} can no longer be achieved. In fact, we will prove in Section~\ref{section:prooffixed} that the achievable region reduces to the trivially achievable region where $R_i\ge H(X_i)$ for all $i$ when $m=3$, though it is nontrivial for $m>3$. For example, suppose $m=4$ and $\Hmsc=\Hmsc_1$. In this case, the achievable region is similar to that given by \eqref{eq:3frregion}, but with an additional sensor. That is, each of the 6 pairs of rates must fall into the corresponding Slepian-Wolf region. In this case, we do three simultaneous Slepian-Wolf codes for each sensor, construct three estimates, each associated with one of the other sensors. For an honest sensor, only one of the other sensors could be a traitor, so at least two of these estimates must be correct. Thus we need only take the plurality of the three estimates to obtain the correct estimate.

\section{Variable-Rate Model and Result}\label{section:model}

\subsection{Notation}

Let $X_i$ be the random variable revealed to sensor $i$, $\Xcal_i$ the alphabet of that variable, and $x_i$ a corresponding realization. A sequence of random variables revealed to sensor $i$ over $n$ timeslots is denoted $X_i^n$, and a realization of it $x_i^n\in\Xcal_i^n$. Let $\Mcal\triangleq\{1,\cdots,m\}$. For a set $\Scal\subset\Mcal$, let $X_\Scal$ be the set of random variables $\{X_i\}_{i\in \Scal}$, and define $x_\Scal$ and $\Xcal_\Scal$ similarly. By $\Scal^c$ we mean $\Mcal\backslash\Scal$. Let $T_\eps^n(X_\Scal)[q]$ be the strongly typical set with respect to the distribution $q$, or the source distribution $p$ if unspecified. Similarly, $H_q(X_\Scal)$ is the entropy with respect to the distribution $q$, or $p$ if unspecified.

\subsection{Communication Protocol}

The transmission protocol is composed of $L$ transactions. In each transaction, the decoder selects a sensor to receive information from and selects which of $K$ encoding functions it should use. The sensor then responds by executing that encoding function and transmitting its output back to the decoder, which then uses the new information to begin the next transaction. 


For each sensor $i\in\Mcal$ and encoding function $j\in\{1,\cdots,K\}$, there is an associated rate $R_{i,j}$. On the $l$th transaction, let $i_l$ be the sensor and $j_l$ the encoding function chosen by the decoder, and let $h_l$ be the number of $l'\in\{1,\cdots,l-1\}$ such that $i_{l'}=i_l$. That is, $h_l$ is the number of times $i_l$ has transmitted prior to the $l$th transaction. Note that $i_l,j_l,h_l$ are random variables, since they are chosen by the decoder based on messages it has received, which depend on the source values. The $j$th encoding function for sensor $i$ is given by
\beq f_{i,j}:\Xcal_i^n\times\Zcal\times\{1,\cdots,K\}^{h_l}\to\{1,\cdots,2^{nR_{i,j}}\}\label{eq:encoding}\eeq
where $\Zcal$ represents randomness generated at the sensor. Let $I_l\in\{1,\cdots,2^{nR_{i_l,j_l}}\}$ be the message received by the decoder in the $l$th transaction. If $i_l$ is honest, then $I_l=f_{i_l,j_l}(X_{i_l}^n,\rho_{i_l},J_l)$, where $\rho_{i_l}\in\Zcal$ is the randomness from sensor $i_l$ and $J_l\in\{1,\cdots,K\}^{h_l}$ is the history of encoding functions used by sensor $i_l$ so far. If $i_l$ is a traitor, however, it may choose $I_l$ based on $W^n$ and it may have any amount of access to previous transmissions $I_1,\cdots,I_{l-1}$ and polling history $i_1,\cdots,i_{l-1}$ and $j_1,\cdots,j_{l-1}$. But, it does not have access to the randomness $\rho_i$ for any honest sensor $i$. Note again that the amount of traitor eavesdropping ability has no effect on achievable rates.

After the decoder receives $I_l$, if $l<L$ it uses $I_1,\cdots,I_l$ to choose the next sensor $i_{l+1}$ and its encoding function index $j_{l+1}$. After the $L$th transaction, it decodes according to the decoding function
\[g:\prod_{l=1}^L \{1,\cdots,2^{nR_{i_l,j_l}}\}\to \Xcal_1^n\times\cdots\times\Xcal_m^n.\]
Note that we impose no restriction whatsoever on the size of the total number of transactions $L$. Thus, a code could have arbitrary complexity in terms of the number of messages passed between the sensors and the decoder. However, in our below definition of achievability, we require that the communication rate from sensors to decoder always exceeds that from decoder to sensors. Therefore while the number of messages may be very large, the amount of feedback is dinimishingly small.

\subsection{Variable-Rate Problem Statement and Main Result}\label{subsection:statement}

Let $\Hcal\subset\Mcal$ be the set of honest sensors. Define the probability of error
\[P_e\triangleq\Pr\big(X_\Hcal^n\ne \hat{X}_\Hcal^n\big)\]
where $(\hat{X}_1^n,\cdots,\hat{X}_m^n)=g(I_1,\cdots,I_L)$. The probability of error will in general depend on the actions of the traitors. Note again that we only require small probability of error on the source estimates corresponding to the honest sensors.

We define a rate function $R(\Hcal,r)$ defined for $\Hcal\in\Hmsc$ and $r\in \Rcal(\Hcal)$ to be \emph{$\alpha$-achievable} if there exists a code such that, for all pairs $(\Hcal,r)$ and any choice of actions by the traitors, $P_e\le\alpha$,
\[\Pr\bigg(\sum_{l=1}^L R_{i_l,j_l}\le R(\Hcal,r)\bigg)\ge 1-\alpha\]
and $\log K\le\alpha nR_{i,j}$ for all $i,j$. This last condition requires, as discussed above, that the feedback rate from the decoder back to the sensors is arbitrarily small compared to the forward rate. A rate function $R(\Hcal,r)$ is \emph{achievable} if for all $\alpha>0$, there is a sequence of $\alpha$-achievable rate functions $\{R'_k(\Hcal,r)\}_{k=1}^\infty$ such that
\[\lim_{k\to\infty}R'_k(\Hcal,r)=R(\Hcal,r).\]
Note that we do not require uniform convergence.

The following definitions allow us to state our main variable-rate result. For any $\Hcal\in\Hmsc$ and $r\in\Rcal(\Hcal)$, let
\[\tilde{r}(w|x_\Hcal)\triangleq\sum_{x_{\Hcal^c}\in\Xmsc_{\Hcal^c}}p(x_{\Hcal^c}|x_\Hcal)r(w|x_\Hcal x_{\Hcal^c}).\]
The extent to which $W$ provides information about $X_{\Hcal^c}$ is irrelevant to the traitors, since all that really matters to the traitors is generating information that appears to agree with $X_{\Hcal}$ as reported by the honest sensors. Thus it will usually be more convenient to work with $\tilde{r}$ rather than $r$. For any $\Scal\in\Hmsc$ and $r'\in \Rcal(\Scal)$, let
\[\Qcal_{\Scal,r'}\triangleq\bigg\{p(x_\Scal)\sum_w\tilde{r}'(w|x_\Scal)\bar{q}(x_{\Scal^c}|w):\forall\bar{q}(x_{\Scal^c}|w)\bigg\}.\]
If $\Scal^c$ were the traitors and $W$ were distributed according to $r'$, $\Qcal_{\Scal,r'}$ is the set of distributions $q$ to which the traitors would have access. That is, if they simulate the proper $\bar{q}(x_{\Scal^c}|w)$ from their received $W$ and combine the result with the actual value of $x_\Scal$, the combination is distributed according to $q$. For any $\Vmsc\subset\Hmsc$, define
\[\Qcal(\Vmsc)\triangleq \bigcap_{\Scal\in \Vmsc}\bigcup_{r'\in \Rcal(\Scal)}\Qcal_{\Scal,r'}.\]
That is, for some distribution $q\in\Qcal(\Vmsc)$, for every $\Scal\in\Vmsc$, if the traitors were $\Scal^c$, they would have access to $q$ for some $r'\in\Rcal(\Scal)$. Thus any distribution in $\Qcal(\Vmsc)$ makes it look to the decoder like any $\Scal\in\Vmsc$ could be the set of honest sensors, so any sensor in $\Ucal(\Vmsc)\triangleq \bigcup_{\Scal\in \Vmsc}\Scal$ is potentially honest.

\begin{theorem}\label{thm:variablerate}
A rate function $R(\Hcal,r)$ is achievable if and only if, for all $(\Hcal,r)$,
\beq R(\Hcal,r)\ge R^*(\Hcal,r)\triangleq \sup_{\Vmsc\subset\Hmsc,\ q\in \Qcal_{\Hcal,r}\cap \Qcal(\Vmsc)}H_q(X_{\Ucal(\Vmsc)}).\label{eq:thm}\eeq
\end{theorem}
See Section~\ref{section:proof} for the proof.

\section{Properties of the Variable-Rate Region}\label{section:comments}

It might at first appear that \eqref{eq:corollary} does not agree with \eqref{eq:3rate}. We discuss several ways in which \eqref{eq:thm}  and \eqref{eq:corollary} can be made more manageable, particularly in the case of perfect traitor information, and show that the two are in fact identical. Let $R^*$ be the minimum rate achievable over all $\Hcal\in\Hmsc$ and $r\in\Rcal(\Hcal)$. Thus by \eqref{eq:thm}, we can write
\beq R^*=\sup_{\Hcal\in\Hmsc,r\in\Rcal(\Hcal)}R^*(\Hcal,r)=\sup_{\Vmsc\subset\Hmsc,\ q\in \Qcal(\Vmsc)}H_q(X_{\Ucal(\Vmsc)}).\label{eq:corollary}\eeq
This is the quantity that appears in \eqref{eq:3rate}. Note also that for perfect traitor information,
\beq\Qcal_{\Scal,r'}=\{q(x_\Mcal):q(x_{\Scal})=p(x_{\Scal})\}.\label{eq:ptiq}\eeq
This means that $\Qcal_{\Hcal,r}\cap \Qcal(\Vmsc)=\Qcal(\Vmsc\cup\{\Hcal\})$. Therefore \eqref{eq:thm} becomes
\[R^*(\Hcal,r)=\sup_{\Vmsc\subset\Hmsc:\Hcal\in\Vmsc,\ q\in\Qcal(\Vmsc)}H_q(X_{\Ucal(\Vmsc)}).\]
The following lemma simplifies calculation of expressions of the form $\sup_{q\in\Qcal(\Vmsc)}H_q(X_{\Ucal(\Vmsc)})$.
\begin{lemma}\label{lemma:form}
Suppose the traitors have perfect information. For any $\Vmsc\subset\Hmsc$, the expression
\beq\sup_{q\in \Qcal(\Vmsc)}H_q(X_{\Ucal(\Vmsc)})\label{eq:lemmaform}\eeq
is maximized by a $q$ satisfying \eqref{eq:ptiq} for all $\Scal\in\Vmsc$ such that, for some set of functions $\{\sigma_\Scal\}_{\Scal\in\Vmsc}$,
\beq q(x_1\cdots x_m)=\prod_{\Scal\in\Vmsc}\sigma_\Scal(x_\Scal).\label{eq:lemmaform2}\eeq
\end{lemma}
\begin{proof}
By \eqref{eq:ptiq}, we need to maximize $H_q(X_{\Ucal(\Vmsc)})$ subject to the constraints that for each $\Scal\in\Vmsc$ and all $x_\Scal\in\Xcal_\Scal$, $q(x_\Scal)=p(x_\Scal)$. This amounts to maximizing the Lagrangian
\begin{multline*}
\Lambda=-\sum_{x_{\Ucal(\Vmsc)}\in\Xcal_{\Ucal(\Vmsc)}}q(x_{\Ucal(\Vmsc)})\log q(x_{\Ucal(\Vmsc)})\\
+\sum_{\Scal\in\Vmsc}\sum_{x_\Scal\in\Xcal_\Scal}\lambda_\Scal(x_\Scal)\big(q(x_\Scal)-p(x_\Scal)\big).
\end{multline*}
Note that for any $\Scal\subset\Ucal(\Vmsc)$,
\[\frac{\partial q(x_\Scal)}{\partial q(x_{\Ucal(\Vmsc)})}=1.\]
Thus, differentiating with respect to $q(x_{\Ucal(\Vmsc)})$ gives, assuming the $\log$ is a natural logarithm,
\begin{align*}
\frac{\partial\Lambda}{\partial q(x_{\Ucal(\Vmsc)})}
=&-\log q(x_{\Ucal(\Vmsc)})-1+\sum_{\Scal\in\Vmsc}\lambda_\Scal(x_\Scal).
\end{align*}
Setting this to 0 gives
\[q(x_{\Ucal(\Vmsc)})=\exp\bigg(-1+\sum_{\Scal\in\Vmsc}\lambda_\Scal(x_\Scal)\bigg)=|\Xcal_{\Ucal(\Vmsc)^c}|\prod_{\Scal\in\Vmsc}\sigma_\Scal(x_\Scal)\]
for some set of functions $\{\sigma_\Scal\}_{\Scal\in\Vmsc}$. Therefore setting 
\[q(x_1\cdots x_m)=\frac{q(x_{\Ucal(\Vmsc)})}{|\Xcal_{\Ucal(\Vmsc)^c}|}\]
satisfies \eqref{eq:lemmaform2}, so if $\sigma_\Scal$ are such that \eqref{eq:ptiq} is satisfied for all $\Scal\in\Vmsc$, $q$ will maximize $H_q(X_{\Ucal(\Vmsc)})$.
\end{proof}

Suppose $m=3$ and $\Hmsc=\Hmsc_1$. If $\Vmsc=\{\{1,2\},\{2,3\}\}$, then $\tilde{q}(x_1x_2x_3)=p(x_1x_2)p(x_3|x_2)$ is in $\Qcal(\Vmsc)$ and by Lemma~\ref{lemma:form} maximizes $H_q(X_1X_2X_3)$ over all $q\in \Qcal(\Vmsc)$. Thus
\begin{align*}\sup_{q\in \Qcal(\Vmsc)}H_q(X_1X_2X_3)&=H_{\tilde{q}}(X_1X_2X_3)\\&=H(X_1X_2X_3)+I(X_1;X_3|X_2).\end{align*}
By similar reasoning, considering $\Vmsc=\{\{1,2\},\{1,3\}\}$ and $\Vmsc=\{\{1,3\},\{2,3\}\}$ results in \eqref{eq:3rate}. Note that if $\Vmsc_1\subset\Vmsc_2$, then $\Qcal(\Vmsc_1)\supset\Qcal(\Vmsc_2)$, so $\Vmsc_2$ need not be considered in evaluating $\eqref{eq:thm}$. Thus we have ignored larger subsets of $\Hmsc_1$, since the value they give would be no greater than the others.

We can generalize to any collection $\Vmsc$ of the form $\{\{\Scal_1,\Scal_2\},\{\Scal_1,\Scal_3\},\cdots,\{\Scal_1,\Scal_k\}\}$, in which case
\[\sup_{q\in\Qcal(\Vmsc)}=H(X_{\Scal_1}X_{\Scal_2})+H(X_{\Scal_3}|X_{\Scal_1})+\cdots+H(X_{\Scal_k}|X_{\Scal_1}).\]
Employing this, we can rewrite \eqref{eq:corollary} for $\Hmsc=\Hmsc_t$ and certain values of $t$. For $t=1$, it becomes
\[R^*= H(X_1\cdots X_m)+\max_{i,i'\in\Mcal}I(X_i;X_{i'}|X_{\{i,i'\}^c}).\]
Again, relative to the Slepian-Wolf result, we always pay a conditional mutual information penalty for a single traitor. For $t=2$,
\begin{multline*}
R^*= H(X_1\cdots X_m)\\+\max\left\{\max_{\Scal,\Scal'\subset\Mcal:|\Scal|=|\Scal'|=2}I(X_\Scal;X_{\Scal'}|X_{(\Scal\cup \Scal')^c}),\right.\\
\left.\max_{i,i',i''\in\Mcal}I(X_i;X_{i'};X_{i''}|X_{\{i,i',i''\}^c})\right\}
\end{multline*}
where $I(X;Y;Z|W)=H(X|W)+H(Y|W)+H(Z|W)-H(XYZ|W)$. For $t=m-1$, $R^*$ is given by \eqref{eq:tm1}. There is a similar formulation for $t=m-2$, though it is more difficult to write down for arbitrary $m$. 

With all these expressions made up of nothing but entropies and mutual informations, it might seem hopeful that \eqref{eq:lemmaform} can be reduced to such an analytic expression for all $\Vmsc$. However, this is not the case. For example, consider $\Vmsc=\{\{1,2,3\},\{3,4,5\},\{5,6,1\}\}$. This $\Vmsc$ is irreducible in the sense that there is no subset $\Vmsc'$ that still satisfies $\Ucal(\Vmsc')=\{1,\cdots,6\}$, but there is no simple distribution $q\in \Qcal(\Vmsc)$ made up of marginals of $p$ that satisfies Lemma \ref{lemma:form}, so it must be found numerically. Still, Lemma \ref{lemma:form} simplifies the calculation considerably.

\section{Proof of Theorem~\ref{thm:variablerate}}\label{section:proof}

\subsection{Converse}

We first show the converse. Fix $\Hcal\in\Hmsc$ and $r\in\Rcal(\Hcal)$. Take any $\Vmsc\subset\Hmsc$, and any distribution $q\in\Qcal_{\Hcal,r}\cap\Qcal(\Vmsc)$. Since $q\in\Qcal_{\Hcal,r}$, there is some $\bar{q}(x_\Tcal|w)$ such that $X_\Hcal$ and $X_\Tcal$ are distributed according to $q$. Since also $q\in\Qcal_{\Scal,r'}$ for all $\Scal\in\Vmsc$ and some $r'\in\Rcal(\Scal)$, if the traitors simulate this $\bar{q}$ and act honestly with these fabricated source values, the decoder will not be able to determine which of the sets in $\Vmsc$ is the actual set of honest sensors. Thus, the decoder must perfectly decode the sources from all sensors in $\Ucal(\Vmsc)$, so if $R(\Hcal,r)$ is a precisely $\alpha$-achievable rate function, $R(\Hcal,r)\ge H_q(X_{\Ucal(\Vmsc)})$.

\subsection{Achievability Preliminaries}

Now we prove achievability. To do so, we will first need the theory of types. Given $y^n\in\Ycal^n$, let $t(y^n)$ be the type of $y^n$. Given a type $t$ with denominator $n$, let $\Lambda_t^n(Y)$ be the set of all sequences in $\Ycal^n$ with type $t$. If $t$ is a joint $y,z$ type with denominator $n$, then let $\Lambda_t^n(Y|z^n)$ be the set of sequences $y^n\in\Ycal^n$ such that $(y^nz^n)$ have joint type $t$, with the convention that this set is empty if the type of $z^n$ is not the marginal of $t$.

We will also need the following definitions. Given a distribution $q$ on an alphabet $\Ycal$, define the $\eta$-ball of distributions
\[B_{\eta}(q)\triangleq\bigg\{q'(\Ycal):\forall x\in\Ycal:
|q(x)-q'(x)|\le\frac{\eta}{|\Ycal|}\bigg\}.\]
Note that the typical set can be written
\[T_\eps^n(X)=\{x^n:t(x^n)\in B_\eps(p)\}.\]
We define slightly modified versions of the sets of distributions from Section~\ref{subsection:statement} as follows:
\begin{align*}\breve{\Qcal}^\eta_{s,r'} &\triangleq \bigcup_{q\in \Qcal_{s,r'}}B_\eta(q),\\
\breve{\Qcal}^\eta(\Vmsc)&\triangleq \bigcap_{\Scal\in \Vmsc}\bigcup_{r'\in \Rcal(\Scal)}\breve{\Qcal}^\eta_{\Scal,r'}.\end{align*}

Finally, we will need the following lemma.
\begin{lemma}\label{lemma:kl} Given an arbitrary $n$ length distribution $q^n(x^n)$ and a type $t$ with denominator $n$ on $\Xcal$, let $q_i(x)$ be the marginal distribution of $q^n$ at time $i$ and $\bar{q}(x)=\frac{1}{n}\sum_{i=1}^n q_i(x)$. If $X^n$ is distributed according to $q^n$ and $\Pr(X^n\in \Lambda_t^n(X))\ge 2^{-n\zeta}$, then $D(t\|\bar{q})\le\zeta$.\end{lemma}
\begin{proof}
Fix an integer $\tilde{n}$. For $\tilde{i}=1,\cdots,\tilde{n}$, let $X^n(\tilde{i})$ be independently generated from $q^n$. Let $\Gamma$ be the set of types $t^n$ on supersymbols in $\Xcal^n$ with denominator $\tilde{n}$ such that $t^n(x^n)=0$ if $x^n\not\in\Lambda_t^n(X)$. Note that
\[|\Gamma|\le(\tilde{n}+1)^{|\Xcal|^n}.\]
If $X^{n\tilde{n}}=(X^n(1),\cdots,X^n(\tilde{n}))$, then
\begin{align*}\Pr\Big(X^{n\tilde{n}}\in\bigcup_{t^n\in \Gamma}\Lambda_{t^n}^{\tilde{n}}(X^n)\Big)&=\Pr(X^n(\tilde{i})\in\Lambda_t^n(X),\forall \tilde{i})\\
&\ge 2^{-n\tilde{n}\zeta}.\end{align*}
But
\begin{multline*}\Pr\Big(X^{n\tilde{n}}\in\bigcup_{t^n\in T^n}\Lambda_{t^n}^{\tilde{n}}(X^n)\Big)
= \sum_{t^n\in \Gamma}\Pr(X^{n\tilde{n}}\in \Lambda_{t^n}^{\tilde{n}}(X^n)\\
\begin{aligned}&\le \sum_{t^n\in \Gamma}2^{-\tilde{n}D(t^n\|q^n)}\\
&\le (\tilde{n}+1)^{|\Xcal|^n}2^{-\tilde{n}\min_{t^n\in \Gamma} D(t^n\|q^n)}.\end{aligned}\end{multline*}
For any $t^n\in\Gamma$, letting $t_i$ be the marginal type at time $i$ gives $\frac{1}{n}\sum_{i=1}^nt_i=t$. Therefore
\begin{align}\zeta+\frac{1}{n\tilde{n}}|\Xcal|^n\log(\tilde{n}+1)&\ge\min_{t^n\in \Gamma}\frac{1}{n}D(t^n\|q^n)\nonumber\\
&\ge\min_{t^n\in \Gamma}\frac{1}{n}\sum_{i=1}^n D(t_i\|q_i)\label{eq:kl1}\\
&\ge D(t\|\bar{q})\label{kl2}\end{align}
where \eqref{eq:kl1} holds by \cite[Lemma 4.3]{Wyner:75IT} and \eqref{kl2} by convexity of the Kullback-Leibler distance in both arguments. Letting $\tilde{n}$ grow proves the lemma.
\end{proof}

The achievability proof proceeds as follows. \hbox{Section~\ref{subsection:scheme}} describes our proposed coding scheme for the case that traitors cannot eavesdrop. In Section~\ref{subsection:errorprob}, we demonstrate that this coding scheme achieves small probability of error when the traitors have perfect information. Section~\ref{subsection:coderate} shows that the coding scheme achieves the rate function $R^*(\Hcal,r)$. In Section~\ref{subsection:imperfect}, we extend the proof to include the case that the traitors have imperfect information. Finally, Section~\ref{subsection:eavesdropping} gives a modification to the coding scheme that can handle eavesdropping traitors.

\subsection{Coding Scheme Procedure}\label{subsection:scheme}

\emph{1) Random Code Structure:} Fix $\eps>0$. The codebook for sensor $i$ is composed of $CJ_i$ separate encoding functions, where $J_i=\left\lceil\frac{\log|\Xcal_i|}{\eps}\right\rceil$ and $C$ is an integer to be defined later. In particular, for $i=1,\cdots,m$ and $c=1,\cdots,C$, let
\begin{align*}
\tilde{f}_{i,c,1}&:\Xcal_i^n\to\{1,\cdots,2^{n(\eps+\nu)}\},\\
\tilde{f}_{i,c,j}&:\Xcal_i^n\to\{1,\cdots,2^{n\eps}\},\quad j=2,\cdots,J_i
\end{align*}
with $\nu$ to defined later. We put tildes on these functions to distinguish them from the $f$s defined in \eqref{eq:encoding}. The $\tilde{f}$s that we define here are functions we use as pieces of the overall encoding functions $f$. Each one is constructed by a uniform random binning procedure. For a given $i$ and $c$, one can think of $\{\tilde{f}_{i,c,j}\}_j$ as a subcodebook that associates each $x_i^n\in\Xcal_i^n$ with a long sequence of bits split into blocks of length $n(\eps+\nu)$ or $n\eps$. Define composite functions
\[\tilde{F}_{i,c,j}(x_i^n)=(\tilde{f}_{i,c,1}(x_i^n),\cdots,\tilde{f}_{i,c,j}(x_i^n)).\]
We can think of $\tilde{F}_{i,c,j}(x_i^n)$ as an index of one of $2^{n(j\eps+\nu)}$ random bins.

\emph{2) Round Method:} We propose a coding scheme made up of $N$ rounds, with each round composed of $m$ phases. In the $i$th phase, transactions are made entirely with sensor $i$. We denote $x_i^n(I)$ as the $I$th block of $n$ source values, but for convenience, we will not include the index $I$ when it is clear from context. As in the three-sensor example, all transactions in the $I$th round are based only on $X_\Mcal^n(I)$. Thus the total block length is $Nn$.

The procedure for each round is identical except for the variable $\Vmsc(I)$ maintained by the decoder. This represents the collection of sets that could be the set of honest sensors based on information the decoder has received as of the beginning of round $I$. The decoder begins by setting $V(1)=\Hmsc$ and then pares it down at the end of each round based on new information.

\emph{3) Encoding and Decoding Rules:} In the $i$th phase, if $i\in \Ucal(\Vmsc(I))$, the decoder makes a number of transactions with sensor $i$ and produces an estimate $\hat{X}_i^n$ of $X_i^n$. The estimate $\hat{X}_i^n$ is of course a random variable, so as usual the lower case $\hat{x}_i^n$ refers to a realization of this variable. If $i\not\in \Ucal(\Vmsc(I))$, then the decoder has determined that sensor $i$ cannot be honest, so it does not communicate with it and sets $\hat{x}_i^n$ to a null value.

For $i\in \Ucal(\Vmsc(I))$, at the beginning of phase $i$, sensor $i$ randomly selects a $c\in\{1,\cdots,C\}$. In the first transaction, sensor $i$ transmits $(c,\tilde{f}_{i,c,1}(X_i^n))$. As the phase continues, in the $j$th transaction, sensor $i$ transmits $\tilde{f}_{i,c,j}(X_i^n)$.

After each transaction, the decoder decides whether to ask for another transaction based on the following rubric. For any $s\subset\Mcal$ and $\hat{x}_s^n\in\Xcal_s^n$, let
\[T_j(\hat{x}_s^n)\triangleq\{x_i^n:H_{t(\hat{x}_s^nx_i^n)}(X_i|X_s)\le j\eps\}.\]
Note that
\[|T_j(\hat{x}_s^n)|\le(n+1)^{|\Xcal_i\times\Xcal_s|}2^{nj\eps}.\]
Let $s_i\triangleq\{1,\cdots,i\}\cap \Ucal(\Vmsc)$ and $\hat{x}_{s_{i-1}}^n$ be the previously decoded source sequences in this round. After $j$ transactions, the decoder will choose to do another transaction if there are no sequences in $T_j(\hat{x}_{s_{i-1}})$ matching the received value of $\tilde{F}_{i,c,j}$. If there is at least one such sequence, let $\hat{x}_i^n$ be one such sequence. If there are several, the decoder chooses from among them arbitrarily.

\emph{4) Round Conclusion:} At the end of round $I$, the decoder produces $\Vmsc(I+1)$ by setting 
\beq \Vmsc(I+1)=\bigg\{\Scal\in \Vmsc(I):t(\hat{x}_{\Ucal(\Vmsc(I))}^n)\in \bigcup_{r'\in R(\Scal)}\breve{\Qcal}^\eta_{\Scal,r'}\bigg\}\label{eq:vupdate}\eeq
for $\eta$ to be defined such that $\eta\ge\eps$ and $\eta\to 0$ as $\eps\to 0$.

\subsection{Error Probability}\label{subsection:errorprob}

Define the following error events:
\begin{align*}
\Ecal_1(I,i)&\triangleq\{\hat{X}_i^n(I)\ne X_i^n(I)\},\\
\Ecal_2(I)&\triangleq\{\Hcal\not\in \Vmsc(I)\},\\
\Ecal_3(I)&\triangleq\{t(\hat{X}_{\Ucal(\Vmsc)}^n(I))\not\in \breve{\Qcal}^\eta_{\Hcal,r}\}\backslash\bigcup_{i\in \Hcal}\Ecal_1(I,i).
\end{align*}
The total probability of error is
\[P_e=\Pr\left(\bigcup_{I=1}^n\bigcup_{i\in \Hcal}\Ecal_1(I,i)\right).\]
For any sequence of events $\Acal_0,\Acal_1,\cdots,\Acal_N$ with $\Acal_I\subset \Acal_{I+1}$ and $\Pr(\Acal_0)=0$,
\begin{align*}\Pr(\Acal_N)
&=1-\prod_{I=1}^{N}\frac{\Pr(\Acal_{I}^c)}{\Pr(\Acal_{I-1}^c)}
=1-\prod_{I=1}^{N}(1-\Pr(\Acal_{I}|\Acal_{I-1}^c))\\
&\le\sum_{I=1}^{N}\Pr(\Acal_{I}|\Acal_{I-1}^c).\end{align*}
Set $\Acal_I=\Ecal_2(I+1)\cup\bigcup_{i\in \Hcal}\Ecal_1(1,i)\cup\cdots\cup\Ecal_1(I,i)$. This satisfies the conditions, and since $\hat{X}_\Mcal^n(I-1)\to\Vmsc(I)\to\hat{X}_\Mcal^n(I)$ is a Markov chain,
\[\Pr(\Acal_I|\Acal_{I-1}^c)=\Pr\Big(\Ecal_2(I+1)\cup\bigcup_{i\in \Hcal}\Ecal_1(I,i)\Big|\Ecal_2^c(I)\Big).\]
Therefore
\[P_e\le\sum_{I=1}^N\Pr\Big(\Ecal_2(I+1)\cup\bigcup_{i\in \Hcal}\Ecal_1(I,i)\Big|\Ecal_2^c(I)\Big).\]
If $\Hcal\in \Vmsc(I)$ and $t(\hat{X}_{\Ucal(\Vmsc)}^n(I))\in \breve{\Qcal}^\eta_{\Hcal,r}$, then $\Hcal\in \Vmsc(I+1)$. Thus
\begin{align*}\Ecal_2(I+1)\backslash\Ecal_2^c(I)&\subset\{t(\hat{X}_{\Ucal(\Vmsc)}^n(I))\not\in \breve{\Qcal}^\eta_{\Hcal,r}\}\backslash\Ecal_2^c(I)\\
&\subset\Big(\Ecal_3(I)\cup\bigcup_{i\in \Hcal}\Ecal_1(I,i)\Big)\backslash\Ecal_2^c(I)\end{align*}
so
\begin{align}P_e&\le\sum_{I=1}^N\Pr\Big(\Ecal_3(I)\cup\bigcup_{i\in \Hcal}\Ecal_1(I,i)\Big|\Ecal_2^c(I)\Big)\nonumber\\
&\le\sum_{I=1}^N\Pr(\Ecal_3(I)|\Ecal_2^c(I))+\sum_{I=1}^N\sum_{i\in \Hcal}\Pr(\Ecal_1(I,i)|\Ecal_2^c(I)).\label{eq:pe1}\end{align}
We will show that for any $I$,
\beq\Pr(\Ecal_3(I)|\Ecal_2^c(I))\le\frac{\alpha}{2N}.\label{eq:event3}\eeq
If the traitors receive perfect source information, then 
\begin{align*}\Ecal_3(I)&\subset\{\hat{X}_\Hcal^n(I)\not\in T_\eps^n(X_\Hcal)\}\cap\{\hat{X}_i^n(I)= X_i^n(I),\forall i\in \Hcal\}\\
&\subset\{X_\Hcal^n(I)\not\in T_\eps^n(X_\Hcal)\}\end{align*}
meaning \eqref{eq:event3} holds for sufficiently large $n$. Thus \eqref{eq:event3} is only nontrivial if the traitors receive imperfect source information. This case is dealt with in Section~\ref{subsection:imperfect}.

Now consider $\Pr(\Ecal_1(I,i)|\Ecal_2^c(I))$ for honest $i$. Conditioning on $\Ecal_2^c(I)$ ensures that $i\in \Ucal(\Vmsc(I))$ for honest $i$, so $\hat{X}_i^n(I)$ will be non-null. The only remaining way to make an error on $X_i^n$ is if there is some transaction $j$ for which there is a sequence $x_i'^n\in T_j(\hat{X}_{s_{i-1}}^n)$ such that $x_i'^n\ne X_i'^n$ and $\tilde{F}_{i,c,j}$ has the same value for $X_i^n$ and $x_i'^n$. However, $s_{i-1}$ may contain traitors. Indeed, it may be made entirely of traitors. Thus, we have to take into account that $\hat{X}_{s_{i-1}}^n$ may be chosen to ensure the existence of such an erroneous $x_i'^n$.

Let 
\begin{multline*}k_1(x_i^n,\hat{x}_{s_{i-1}}^n)\triangleq|\{c:\exists j,x_i'^n\in T_j(\hat{x}_{s_{i-1}}^n)\backslash\{x_i^n\}:\\F_{i,c,j}(x_i'^n)=F_{i,c,j}(x_i^n)\}|.\end{multline*}
That is, $k_1$ is the number of subcodebooks that if chosen could cause an error. Recall that sensor $i$ chooses the subcodebook randomly from the uniform distribution. Thus, given $x_i^n$ and $\hat{x}_{s_{i-1}}^n$, the probability of an error resulting from a bad choice of subcodebook is $k_1(x_i^n,\hat{x}_{s_{i-1}}^n)/C$. Furthermore, $k_1$ is based strictly on the codebook, we can think of $k_1$ as a random variable based on the codebook choice. Averaging over all possible codebooks,
\[\Pr(\Ecal_1(I,i)|\Ecal_2^c(I))\le \Embb\!\!\sum_{x_i^n\in\Xcal_i^n}p(x_i^n)\!\max_{\hat{x}_{s_{i-1}}^n\in\Xcal_{s_{i-1}}^n}\!\frac{k_1(x_i^n,\hat{x}_{s_{i-1}}^n)}{C}\]
where the expectation is taken over all codebooks.

Let $\Ccal$ be the set of all codebooks. We define a subset $\Ccal_1$, and show that the probability of error can be easily bounded for any codebook in $\Ccal\backslash\Ccal_1$, and the probability of a codebook being chosen in $\Ccal_1$ is small. In particular, let $\Ccal_1$ be the set of codebooks for which, for any $x_i^n\in\Xcal_i^n$ and $\hat{x}_{s_{i-1}}^n\in\Xcal_{s_{i-1}}^n$, $k_1(x_i^n,\hat{x}_{s_{i-1}}^n)> B$, for an integer $B\le C$ to be defined later. Then
\begin{align}
\Pr(\Ecal_1(I,i)|\Ecal_2^c(I))&\le\Pr(\Ccal\backslash\Ccal_1)\sum_{x_i^n\in\Xcal_i^n}p(x_i^n)\max_{\hat{x}_{s_{i-1}}^n\in\Xcal_{s_{i-1}}^n}\frac{B}{C}\nonumber\\
&+\Pr(\Ccal_1)\sum_{x_i^n\in\Xcal_i^n}p(x_i^n)\max_{\hat{x}_{s_{i-1}}^n\in\Xcal_{s_{i-1}}^n}\frac{C}{C}\nonumber\\
&\le \frac{B}{C}+\Pr(\Ccal_1).\label{eq:pe2}
\end{align}

Since each subcodebook is generated identically, $k_1$ is a binomial random variable with $C$ trials and probability of success
\begin{align*}
P&\triangleq\Pr\!\big(\exists j,x'^n_i\in T_j(\hat{x}_{s_{i-1}}^n)\backslash\{x_i^n\}:F_{i,c,j}(x'^n_i)=F_{i,c,j}(x_i^n)\big)\\
&\le \sum_j\sum_{x'^n_i\in T_j(\hat{x}_{s_{i-1}}^n)\backslash\{x_i^n\}}\Pr\big(F_{i,c,j}(x'^n_i)=F_{i,c,j}(x_i^n)\big)\\
&\le J_i\left|T_j(\hat{x}_{s_{i-1}}^n)\right|2^{-n(j\eps+\nu)}\\
&\le J_i (n+1)^{|\Xcal_i\times\Xcal_{s_{i-1}}|}2^{-n\nu}\le 2^{n(\eps-\nu)}
\end{align*}
for sufficiently large $n$. For a binomial random variable $X$ with mean $\bar{X}$ and any $\kappa$, we can use the Chernoff bound to write
\beq\Pr(X\ge \kappa)\le \left(\frac{e\bar{X}}{\kappa}\right)^{\kappa}.\label{eq:binomialbound}\eeq
Therefore
\[\Pr(k_1(x_i^n,\hat{x}_{s_{i-1}}^n)>B)\le\left(\frac{eCP}{B+1}\right)^{B+1}\le 2^{nB(\eps-\nu)}\]
if $\nu>\eps$ and $n$ is sufficiently large. Thus
\begin{align}
\Pr(\Ccal_1)&=\Pr(\exists x_i^n,\hat{x}_{s_{i-1}}^n:k_1(x_i^n,\hat{x}_{s_{i-1}}^n)>B)\nonumber\\
&\le\sum_{x_i^n}\sum_{\hat{x}_{s_{i-1}}^n}\Pr(k(x_i^n,\hat{x}_{s_{i-1}}^n)>B)\nonumber\\
&\le \sum_{x_i^n}\sum_{\hat{x}_{s_{i-1}}^n}2^{nB(\eps-\nu)}\nonumber\\
&=2^{n(\log|\Xcal_i|+\log|\Xcal_{s_{i-1}}|+B(\eps-\nu)}.\label{eq:pe3}
\end{align}
Combining \eqref{eq:pe1} with \eqref{eq:event3}, \eqref{eq:pe2}, and \eqref{eq:pe3} gives
\begin{align*}
P_e&\le \frac{\alpha}{2}+\sum_{I=1}^N\sum_{i\in \Hcal}\left(\frac{B}{C}+2^{n(\log|\Xcal_i|+\log|\Xcal_{s_{i-1}}|+B(\eps-\nu)}\right)\\
&\le\frac{\alpha}{2}+Nm\left(\frac{B}{C}+2^{n(\log|\Xcal_\Mcal|+B(\eps-\nu))}\right)\end{align*}
which is less than $\alpha$ for sufficiently large $n$ if
\[B>\frac{\log|\Xcal_\Mcal|}{\nu-\eps}\]
and
\[C\ge\frac{3NmB}{\alpha}>\frac{3Nm\log|\Xcal_\Mcal|}{\alpha(\nu-\eps)}.\]

\subsection{Code Rate}\label{subsection:coderate}

The discussion above placed a lower bound on $C$. However, for sufficiently large $n$, we can make 
$\frac{1}{n}\log C\le\eps,$ meaning it takes no more than $\eps$ rate to transmit the subcodebook index $c$. Therefore the rate for phase $i$ is at most $(j+1)\eps+\nu$, where $j$ is the number of transactions in phase $i$. Transaction $j$ must be the earliest one with $\hat{x}_i^n\in T_j(\hat{x}_{s_{i-1}})$, otherwise it would have been decoded earlier. Thus $j$ is the smallest integer for which
\[H_{t(\hat{x}_{s_{i-1}}^n\hat{x}_i^n)}(X_i|X_{s_{i-1}})\le j\eps\]
meaning
\[j\eps\le H_{t(\hat{x}_{s_{i-1}}^n\hat{x}_i^n)}(X_i|X_{s_{i-1}})+\eps.\]
By \eqref{eq:vupdate}, for all $s\in \Vmsc(I+1)$, $t(\hat{x}_{\Ucal(\Vmsc(I))}^n)\in \bigcup_{r'\in R(s)}\breve{\Qcal}^\eta_{s,r'}$, meaning 
\beq t(\hat{x}_{\Ucal(\Vmsc(I))})\in\bigcap_{s\in \Vmsc(I+1)}\bigcup_{r'\in R(s)}\breve{\Qcal}^\eta_{s,r'}=\breve{\Qcal}^\eta(\Vmsc(I+1)).\label{eq:typeproperty}\eeq
Combining this with \eqref{eq:event3}, with probability at least $1-\alpha$, $t(\hat{x}_{\Ucal(\Vmsc(I))})\in \breve{\Qcal}^\eta_{\Hcal,r}\cap \breve{\Qcal}^\eta(\Vmsc(I+1))$. Therefore with high probability the rate for all of round $I$ is at most
\begin{align}
&\sum_{i\in \Ucal(\Vmsc(I))} \left(H_{t(\hat{x}_{s_{i-1}}^n\hat{x}_i^n)}(X_i|X_{s_{i-1}})+2\eps+\nu\right)\nonumber\\
&\le H_{t(\hat{x}_{\Ucal(\Vmsc(I))})}\left(X_{\Ucal(\Vmsc)}\right)+m(2\eps+\nu)\nonumber\\
&\le \sup_{q\in \breve{\Qcal}^\eta_{\Hcal,r}\cap \breve{\Qcal}^\eta(\Vmsc(I+1))}H_q\left(X_{\Ucal(\Vmsc)}\right)+m(2\eps+\nu)\nonumber\\
&\le \sup_{q\in \breve{\Qcal}^\eta_{\Hcal,r}\cap \breve{\Qcal}^\eta(\Vmsc(I+1))}H_q\left(X_{\Ucal(\Vmsc(I+1))}\right)\nonumber\\
&\qquad\qquad+\log \left|\Xcal_{\Ucal(\Vmsc(I))\backslash \Ucal(\Vmsc(I+1))}\right|+m(2\eps+\nu)\nonumber\\
&\le \sup_{\Vmsc\subset\Hmsc,\ q\in \breve{\Qcal}^\eta_{\Hcal,r}\cap \breve{\Qcal}^\eta(\Vmsc)}H_q(X_{\Ucal(\Vmsc)})\nonumber\\
&\qquad\qquad+\log \left|\Xcal_{\Ucal(\Vmsc(I))\backslash \Ucal(\Vmsc(I+1))}\right|+m(2\eps+\nu).\label{eq:rate4}
\end{align}
Whenever $\Ucal(\Vmsc(I))\backslash \Ucal(\Vmsc(I+1))\ne\emptyset$, at least one sensor is eliminated. Therefore the second term in \eqref{eq:rate4} will be nonzero in all but at most $m$ rounds. Moreover, although we have needed to bound $\nu$ from below, we can still choose it such that $\nu\to 0$ as $\eps\to 0$. Thus if $N$ is large enough, the rate averaged over all rounds is no more than 
\[R_\eps(\Hcal,r)\triangleq\sup_{\Vmsc\subset\Hmsc,\ q\in \breve{\Qcal}^\eta_{\Hcal,r}\cap \breve{\Qcal}^\eta(\Vmsc)}H_q(X_{\Ucal(\Vmsc)})+\dot{\eps}\]
where $\dot{\eps}\to 0$ as $\eps\to 0$. This is a precisely $\alpha$-achievable rate function. By continuity of entropy,
\[\lim_{\eps\to 0}R_\eps(\Hcal,r)=\sup_{\Vmsc\subset\Hmsc,\ q\in \Qcal_{\Hcal,r}\cap \Qcal(\Vmsc)}H_q(X_{\Ucal(\Vmsc)})=R^*(\Hcal,r)\]
so $R^*(\Hcal,r)$ is achievable.

\subsection{Imperfect Traitor Information}\label{subsection:imperfect}

We now consider the case that the traitors have access to imperfect information about the sources. The additional required piece of analysis is to prove \eqref{eq:event3}. That is
\[\Pr(t(\hat{X}_{\Ucal(\Vmsc(I))}^n(I))\not\in\breve{\Qcal}^\eta_{\Hcal,r},\hat{x}_\Hcal=X_\Hcal|\Hcal\in \Vmsc(I))\le \frac{\alpha}{2N}.\]
We will in fact prove the slightly stronger statement
\begin{multline}\Pr(t(X_{\Hcal\cap \Ucal(\Vmsc(I))}^n(I)\hat{X}_{\Tcal\cap \Ucal(\Vmsc(I))}^n(I))\not\in\breve{\Qcal}^\eta_{\Hcal,r}|\Hcal\in \Vmsc(I))\\
\le \frac{\alpha}{2N}.\label{eq:toprove}\end{multline}
Since we condition on $\Hcal\in \Vmsc(I)$, we can assume $\Hcal\subset \Ucal(\Vmsc(I))$. For notational convenience, let $Y=X_\Hcal(I)$ and $Z=X_{\Tcal\cap \Ucal(\Vmsc(I))}(I)$, so \eqref{eq:toprove} becomes
\[\Pr(t(Y^n\hat{Z}^n)\not\in\breve{\Qcal}^\eta_{\Hcal,r}|\Hcal\in \Vmsc(I))\le \frac{\alpha}{2N}.\]

\renewcommand{\j}{\mathbf{j}}

Based on their received value of $W^n$, the traitors choose a value of $c$ and then a series of messages for each traitor in $\Ucal(\Vmsc(I))$. The number of messages each traitor actually gets to send depends on how long it takes for the decoder to construct a source estimate. Let $\j=\{j_i\}_{i\in T\cap \Ucal(\Vmsc(I))}$ be a vector representing the number of transactions that take place with each traitor in $\Ucal(\Vmsc(I))$. There are $J_\Tcal\triangleq\prod_{i\in\Tcal\cap \Ucal(\Vmsc(I))}J_i$ different possible values of $\j$. We can think of any series of values of $c$ and messages as a bin (i.e. a subset $\Zcal^n$); that is, all sequences that map to the same messages in the subcodebooks denoted by the values of $c$. Let $R(\j)$ be the rate at which the traitors transmit given $\j$. Thus if we let $\Bcal_R$ be the set of all bins in the codebook constructed at rate $R$, the traitors are equivalent to a group of potentially random functions $g_{\j}:\Wcal^n\to \Bcal_{R(\j)}$.

Consider a joint $y,z$ type $t$. In order for $(Y^n\hat{z}^n)$ to have type $t$ for a given $\j$, we need $R(\j)\ge H_t(Z|Y)+\nu$. Thus
\begin{multline*}\Pr((Y^n\hat{z}^n)\in\Lambda_t^n(YZ))\le\Pr(\exists \j:R(\j)\ge H_t(Z|Y)+\nu,\\z^n\in g_{\j}(W^n)\cap\Lambda_t^n(Z|Y^n)).\end{multline*}
Let $\delta\triangleq\frac{\eps}{4N}$,
\begin{multline*}\delta_{t,\j} \triangleq \Pr((Y^n,W^n)\in T_\eps^n(YW),\\
\exists z^n\in g_{\j}(W^n)\cap\Lambda_t^n(Z|Y^n))\end{multline*}
and
\[\Pcal\triangleq\left\{t:\max_{\j:R(\j)\ge H_t(Z|Y)+\nu}\delta_{t,\j}\ge \frac{\delta}{(n+1)^{|\Ycal\times\Zcal|}J_T}\right\}.\]
We will show that $\Pcal\subset \breve{\Qcal}^\eta_{\Hcal,r}$, so that
\begin{align*}&\Pr(t(Y^n\hat{z}^n)\not\in\breve{\Qcal}_{\Hcal,r}|\Hcal\in \Vmsc(I))\\
&\le\Pr(t(Y^n\hat{z}^n)\not\in \Pcal|\Hcal\in \Vmsc(I))\\
&\le\Pr\big(\exists t\in \Pcal^c,\j:R(\j)\ge H_t(Z|Y)+\nu,\\
&\qquad\qquad\qquad z^n\in g_{\j}(W^n)\cap\Lambda_t^n(Z|Y^n)\big|\Hcal\in \Vmsc(I)\big)\\
&\le\Pr((Y^n,W^n)\not\in T_\eps^n(YW))+\sum_{t\in \Pcal^c}\sum_{\j:R(\j)\ge H_t(Z|Y)+\nu}\delta_{t,\j}\\
&\le \delta+(n+1)^{|\Ycal\times\Zcal|}J_T\frac{\delta}{(n+1)^{|\Ycal\times\Zcal|}J_T}
= 2\delta=\frac{\alpha}{2N}\end{align*}
for sufficiently large $n$.

Fix $t\in \Pcal$. There is some $\j$ with $R(\j)\ge H_t(Z|Y)+\nu$ and $\delta_{t,\j}\ge\frac{\delta}{(n+1)^{|\Ycal\times\Zcal|}J_T}$. Any random $g_{\j}$ is a probabilistic combination of a number of deterministic functions, so if this lower bound on $\delta_{t,\j}$ holds for a random $g_{\j}$, it must also hold for some deterministic $g_{\j}$. Therefore we do not lose generality to assume from now on that $g_{\j}$ is deterministic. We also drop the $\j$ subscript for convenience. Our method of proof will be to demonstrate that such a functions $g$ can only exist if there is also a $h:\Wcal^n\to\Zcal^n$ with almost the same properties. That is, if the traitors can fabricate a counterfeit bin made up of source sequences, they can fabricate a single counterfeit source sequence contained in this bin that works nearly as well.

Define the following sets:
\begin{multline*}
A_\eps^n(Y|w^n)\triangleq\{y^n\in T_\eps^n(Y|w^n):\\
\shoveright{\exists z^n\in g(w^n)\cap \Lambda_t^n(Z|y^n)\},}\\
\shoveleft{A_\eps^n(W)\triangleq\Big\{w^n\in T_\eps^n(W):}\\
\Pr(Y^n\in A_\eps^n(Y|w^n)|W^n=w^n)\ge \frac{\delta}{2(n+1)^{|\Ycal\times\Zcal|}J_T}\Big\}.\end{multline*}
Applying the definitions of $\Pcal$ and $\delta_{t,\j}$ gives
\begin{align*}
&\frac{\delta}{(n+1)^{|\Ycal\times\Zcal|}J_T}\\
&\le\Pr((Y^nW^n)\in T_\eps^n(YW):\exists z^n\in g(W^n)\cap \Lambda_t^n(Z|Y^n))\\
&=\sum_{w^n\in T_\eps^n(W)}p(w^n)\Pr(Y^n\in A_\eps^n(Y|w^n)|W^n=w^n)\\
&\le \Pr(W^n\in A_\eps^n(W))+\frac{\delta}{2(n+1)^{|\Ycal\times\Zcal|}J_T}
\end{align*}
meaning $\Pr(W^n\in A_\eps^n(W))\ge \frac{\delta}{2(n+1)^{|\Ycal\times\Zcal|}J_T}.$ Fix $w^n\in A_\eps^n(W)$. Since $A_\eps^n(Y|w^n)\subset T_\eps^n(Y|w^n)$,
\[|A_\eps^n(Y|w^n)|\ge \frac{\delta}{2(n+1)^{|\Ycal\times\Zcal|}J_T}2^{n(H(Y|W)-\eps)}.\]
Note also that
\begin{align*}|A_\eps^n(Y|w^n)|&\le\sum_{y^n\in T_\eps^n(Y|w^n)}|g(w^n)\cap \Lambda_t^n(Z|y^n)|\\
&=\sum_{z^n\in g(w^n)}|\Lambda_t^n(Y|z^n)\cap T_\eps^n(Y|w^n)|.\end{align*}
Setting $k_2(z^n,w^n)\triangleq|\Lambda_t^n(Y|z^n)\cap T_\eps^n(Y|w^n)|$, 
\begin{align}\sum_{z^n\in g(w^n)}k_2(z^n,w^n)&\ge \frac{\delta}{2(n+1)^{|\Ycal\times\Zcal|}J_T}2^{n(H(Y|W)-\eps)}\nonumber\\
&\ge 2^{n(H(Y|W)-2\eps)}\label{eq:k2binsum}\end{align}
for sufficiently large $n$. We will show that there is actually a single $\tilde{z}^n\in g(w^n)$ such that $k_2(\tilde{z}^n,w^n)$ represents a large portion of the above sum, so $\tilde{z}^n$ itself is almost as good as the entire bin. Then setting $h(w^n)=\tilde{z}^n$ will give us the properties we need. Note that
\begin{align}
\sum_{z^n\in \Zcal^n}k_2(z^n,w^n)\nonumber
&=\sum_{y^n\in T_\eps^n(Y|w^n)}|\Lambda_t^n(Z|y^n)|\nonumber\\
&\le 2^{n(H(Y|W)+H_t(Z|Y)+\eps)}.\label{eq:k2fullsum}\end{align}
Certainly 
\[k_2(z^n,w^n)\le|T_\eps^n(Y|w^n)|\le 2^{n(H(Y|W)+\eps)}\]
so if we let $l(z^n)$ be the integer such that
\begin{multline}
2^{n(H(Y|W)-l(z^n)\eps)}< k_2(z^n,w^n)\\\le 2^{n(H(Y|W)-(l(z^n)-1)\eps)}.\label{eq:ldef}\end{multline}
then $l(z^n)\ge 0$. Furthermore, if $k_2(z^n,w^n)>0$, then $l(z^n)\le L\triangleq\lceil\frac{H(Y|W)}{\eps}\rceil.$ Let $M(l)=|\{z^n\in \Zcal^n:l(z^n)=l\}|$. Then from \eqref{eq:k2fullsum}, for some $l$,
\begin{align*}
2^{n(H(Y|W)+H_t(Z|Y)+\eps)}&\ge\sum_{z^n\in \Zcal^n}k_2(z^n,w^n)\\
&\ge \sum_{z^n\in \Zcal^n:l(z^n)=l}k_2(z^n,w^n)\\
&\ge M(l)2^{n(H(Y|W)-l\eps)}
\end{align*}
giving
\[M(l)\le 2^{n(H_t(Z|Y)+(l+1)\eps)}.\]
For any bin $b\in\Bcal_{R(\j)}$, let $\tilde{M}(l,b)\triangleq|\{z^n\in b:l(z^n)=l\}|$. Since $R(\j)\ge H_t(Z|Y)+\nu$, $\tilde{M}(l,b)$ is a binomial random variable with $M(l)$ trials and probability of success at most $2^{-n(H_t(Z|Y)+\nu)}$. Thus
\begin{align*}\Embb\tilde{M}(l,b)&\le 2^{n(H_t(Z|Y)+(l+1)\eps)}2^{-n(H_t(Z|Y)+\nu)}\\&=2^{n((l+1)\eps-\nu)}.\end{align*}
Let $\Ccal_2$ be the set of codebooks such that for any group of sensors, subcodebooks, type $t$, transactions $\j$, sequence $w^n\in\Wcal^n$, bin $b$ and integer $l$, either $\tilde{M}(l,b)\ge 2^{n\eps}$ if $(l+1)\eps-\nu\le 0$ or $\tilde{M}(l,b)\ge 2^{n((l+2)\eps-\nu)}$ if $(l+1)\eps-\nu>0$. We will show that the probability of $\Ccal_2$ is small, so we may disregard it. Again using \eqref{eq:binomialbound}, if $(l+1)\eps-\nu\le 0$,
\[\Pr(\tilde{M}(l,b)\ge 2^{n\eps})
\le \left(\frac{e}{2^{n(-l\eps+\nu)}}\right)^{2^{n\eps}}\le 2^{-2^{n\eps}}\]
and if $(l+1)\eps-\nu>0$,
\begin{align*}\Pr(\tilde{M}(l,b)\ge 2^{n((l+2)\eps-\nu)})
&\le \left(\frac{e}{2^{n\eps}}\right)^{2^{n((l+2)\eps-\nu)}}\\
&\le 2^{-2^{n((l+2)\eps-\nu)}}\end{align*}
both for sufficiently large $n$.  Therefore
\begin{multline*}\Pr(\Ccal_2)\le 2^m C^m (n+1)^{|\Xcal_\Mcal|} J_1\cdots J_m |\Wcal|^n 2^{n(|\Xcal_\Mcal|+\nu)}\\ 
\cdot\left(\sum_{0\le l\le\frac{\nu}{\eps}-1}2^{-2^{n\eps}}+\sum_{\frac{\nu}{\eps}-1<l\le L}2^{-2^{n((l+2)\eps-\nu)}}\right)\end{multline*}
which vanishes as $n$ grows.

We assume from now on that the codebook is not in $\Ccal_2$, meaning in particular that $\tilde{M}(l,g(w^n))\le 2^{n\eps}$ for $(l+1)\eps-\nu\le 0$ and $\tilde{M}(l,g(w^n))\le 2^{n((l+2)\eps-\nu)}$ for $(l+1)\eps-\nu>0$. Applying these and \eqref{eq:ldef} to \eqref{eq:k2binsum} and letting $\tilde{l}$ be an integer defined later,
\begin{align*}
2^{-n2\eps}&\le 2^{-nH(Y|W)}\sum_{z^n\in g(w^n)}k_2(z^n,w^n)\\
&\le \sum_{l=0}^L \tilde{M}(l,g(w^n))2^{-n(l-1)\eps}\\
&=\sum_{0\le l<\tilde{l}}\tilde{M}(l,g(w^n))2^{-n(l-1)\eps}\\
&\qquad+\sum_{\tilde{l}\le l\le\frac{\nu}{\eps}-1}\tilde{M}(l,g(w^n))2^{-n(l-1)\eps}\\
&\qquad+\sum_{\frac{\nu}{\eps}-1<l\le L}\tilde{M}(l,g(w^n))2^{-n(l-1)\eps}\\
&\le\sum_{0\le l<\tilde{l}}\tilde{M}(l,g(w^n))2^{n\eps}+\sum_{\tilde{l}\le l\le\frac{\nu}{\eps}-1}2^{n\eps}2^{-n(\tilde{l}-1)\eps}\\
&\qquad+\sum_{\frac{\nu}{\eps}-1<l\le L}2^{n((l+2)\eps-\nu)}2^{-n(l-1)\eps}\\
&\le\sum_{0\le l<\tilde{l}}\tilde{M}(l,g(w^n))2^{n\eps}+L2^{n(-\tilde{l}+2)\eps}+L2^{n(3\eps-\nu)}.
\end{align*}
Therefore
\[\sum_{0\le l<\tilde{l}}\tilde{M}(l,g(w^n))
\ge 2^{-n3\eps}\left(1-L2^{n(-\tilde{l}+4)\eps}-L2^{n(5\eps-\nu)}\right).\]
Setting $\tilde{l}=5$ and $\nu>5\eps$ ensures that the right hand side is positive for sufficiently large $n$, so there is at least one $z^n\in g(w^n)$ with $|T_\eps^n(Y|w^n)\cap \Lambda_t^n(Y|z^n)|\ge 2^{n(H(Y|W)-4\eps)}$. Now we define $h:\Wcal^n\to\Zcal^n$ such that $h(w^n)$ is such a $z^n$ for $w^n\in A_\eps^n(W)$ and $h(w^n)$ is arbitrary for $w^n\not\in A_\eps^n(W)$. If we let $\tilde{Z}^n=h(W^n)$,
\begin{align*}
\Pr((&Y^n\tilde{Z}^n)\in\Lambda_t^n(YZ))\\
&\ge \sum_{w^n\in A_\eps^n(W)}p(w^n)\Pr(Y^n\in \Lambda_t^n(Y|h(w^n))|W^n=w^n)\\
&\ge\sum_{w^n\in A_\eps^n(W)}p(w^n)\\
&\qquad\cdot\Pr(Y^n\in T_\eps^n(Y|w^n)\cap \Lambda_t^n(Y|h(w^n))|W^n=w^n)\\
&\ge\Pr(W^n\in A_\eps^n(W))2^{-n(H(Y|W)+\eps)}2^{n(H(Y|W)-4\eps)}\\
&\ge \frac{\delta}{2(n+1)^{|\Ycal\times\Zcal|}} 2^{-n5\eps}.\end{align*}
The variables $(Y^nW^n\tilde{Z}^n)$ are distributed according to
\[q^n(y^nw^nz^n)=\left(\prod_{i=1}^n p(y_i)r(w_i|y_i)\right)\mathbf{1}\{z^n=h(w^n)\}.\]
Let $q_i(ywz)$ be the marginal distribution of $q^n(y^nw^nz^n)$ at time $i$. It factors as
\[q_i(ywz)=p(y)r(w|y)q_i(z|w).\]
Let $\bar{q}(yz)\triangleq\frac{1}{n}\sum_i q_i(yz)$ and $\bar{q}(z|w)\triangleq\frac{1}{n}\sum_{i=1}^nq_i(z|w)$. Then
\[\bar{q}(yz)=p(y)\sum_{w}r(w|y)\bar{q}(z|w)\]
so by Lemma~\ref{lemma:kl},
\begin{multline*}D\left(t\Big\|p(y)\sum_{w}r(w|y)\bar{q}(z|w)\right)\\
\le-\frac{1}{n}\log\left(\frac{\delta}{2(n+1)^{|\Ycal\times\Zcal|}}\right)+5\eps.\end{multline*}
Therefore $t\in\breve{\Qcal}^\eta_{\Hcal,r}$ for sufficiently large $n$ and some $\eta$ such that $\eta\to 0$ as $\eps\to 0$.

\subsection{Eavesdropping Traitors}\label{subsection:eavesdropping}

We consider now the case that the traitors are able to overhear communication between the honest sensors and the decoder. If the traitors have perfect information, then hearing the messages sent by honest sensors will not give them any additional information, so the above coding scheme still works identically. If the traitors have imperfect information, we need to slightly modify the coding scheme, but the achievable rates are the same.

The important observation is that eavesdropping traitors only have access to messages sent in the past. Thus, by permuting the order in which sensors are polled in each round, the effect of the eavesdropping can be eliminated. In a given round, let $\Hcal'$ be the set of honest sensors that transmit before any traitor. Since the additional information gain from eavesdropping will be no more than the values of $X_{\Hcal'}^n$, the rate for this round, if no sensors are eliminated (i.e. $\Ucal(\Vmsc(I+1))=\Ucal(\Vmsc(I))$), will be no more than the rate without eavesdropping when the traitors have access to $W'^n=(W^n,X_{\Hcal'}^n)$. The goal of permuting the transmission order is to find an ordering in which all the traitors transmit before any of the honest sensors, since then the achieved rate, if no sensors are eliminated, will be the same as with no eavesdropping. It is possible to determine when such an order occurs because it will be the order that produces the smallest rate.

More specifically, we will alter the transmission order from round to round in the following way. We always choose an ordering such that for some $\Scal\in\Vmsc$, the sensors $\Scal^c$ transmit before $\Scal$. We cycle through all such orderings until for each $\Scal$, there has been one round with a corresponding ordering in which no sensors were eliminated. We then choose one $\Scal$ that never produced a rate larger than the smallest rate encountered so far. We perform rounds in a order corresponding to $\Scal$ from then on. If the rate ever changes and is no longer the minimum rate encountered so far, we choose a different minimizing $\Scal$. The minimum rate will always be no greater than the achievable rate without eavesdropping, so after enough rounds, we achieve the same average rate.

\section{Fixed-Rate Coding}\label{section:fixedrate}

Consider an $m$-tuple of rates $(R_1,\cdots,R_m)$, encoding functions
$f_i:\Xcal_i^n\to\{1,\cdots,2^{nR_i}\}$
for $i\in\Mcal$, and decoding function
\[g:\prod_{i=1}^m\{1,\cdots,2^{nR_i}\}\to\Xcal_1^n\times\cdots\times\Xcal_m^n.\]
Let $I_i\in\{1,\cdots,2^{nR_i}\}$ be the message transmitted by sensor $i$. If sensor $i$ is honest, $I_i=f_i(X_i^n)$. If it is a traitor, it may choose $I_i$ arbitrarily, based on  $W^n$. Define the probability of error $P_e\triangleq\Pr\big(X_\Hcal^n\ne \hat{X}_\Hcal^n\big)$
where $\hat{X}_\Mcal^n=g(I_1,\cdots,I_m)$. 

We say an $m$-tuple $(R_1,\cdots,R_m)$ is \emph{deterministic-fixed-rate achievable} if for any $\eps>0$ and sufficiently large $n$, there exist coding functions $f_i$ and $g$ such that, for any choice of actions by the traitors, $P_e\le\eps$. Let $\Rcal_{\text{dfr}}\subset\Rmbb^m$ be the set of deterministic-fixed-rate achievable $m$-tuples.

For randomized fixed-rate coding, the encoding functions become
\[f_i:\Xcal_i^n\times\Zcal\to\{1,\cdots,2^{nR_i}\}\]
where $\Zcal$ is the alphabet for the randomness. If sensor $i$ is honest, $I_i=f_i(X_i^n,\rho_i)$, where $\rho_i\in\Zcal$ is the randomness produced at sensor $i$. Define an $m$-tuple to be \emph{randomized-fixed-rate achievable} in the same way as above, and $\Rcal_{\text{rfr}}\subset\Rmbb^m$ to be the set of randomized-fixed-rate achievable rate vectors.

For any $\Scal\subset\Mcal$, let $\text{SW}(X_\Scal)$ be the Slepian-Wolf rate region on the random variables $X_\Scal$. That is, 
\[\text{SW}(X_\Scal)\triangleq\bigg\{R_\Scal:\forall \Scal'\subset \Scal:\sum_{i\in \Scal'}R_i\ge H(X_{\Scal'}|X_{\Scal\backslash \Scal'})\bigg\}.\]
Let
\begin{align*}\Rcal^*_{\text{rfr}}&\triangleq\{(R_1,\cdots,R_m):\forall \Scal\in\Hmsc:R_\Scal\in\text{SW}(X_\Scal)\},\\
\Rcal^*_{\text{dfr}}&\triangleq\{(R_1,\cdots,R_m)\in\Rcal^*_{\text{rfr}}:\forall \Scal_1,\Scal_2\in\Hmsc:\\
&\qquad\qquad\text{if }\exists r\in R(\Scal_2):H_r(X_{\Scal_1\cap\Scal_2}|W)=0,\\
&\qquad\qquad\text{then }R_{\Scal_1\cap\Scal_2}\in\text{SW}(X_{\Scal_1\cap\Scal_2})\}
\end{align*}
The following theorem gives the rate regions explicitly.
\begin{theorem}\label{thm:fixedrate}
The fixed-rate achievable regions are given by
\[\Rcal_{\text{dfr}}=\Rcal^*_{\text{dfr}}\qquad\text{and}\qquad \Rcal_{\text{rfr}}=\Rcal^*_{\text{rfr}}.\]
\end{theorem}

\section{Proof of Theorem~\ref{thm:fixedrate}}\label{section:prooffixed}

\subsection{Converse for Randomized Coding}

Assume $(R_1,\cdots,R_m)$ is randomized-fixed-rate achievable. Fix $\Scal\in\Hmsc$. Suppose $\Scal^c$ are the traitors and perform a black hole attack. Thus $\hat{X}_\Scal^n$ must be based entirely on $\{f_i(X_i^n)\}_{i\in \Scal}$, and since $\Pr(X_\Scal\ne\hat{X}_\Scal)$ can be made arbitrarily small, by the converse of the Slepian-Wolf theorem, which holds even if the encoders may use randomness, $R_\Scal\in \text{SW}(X_\Scal)$.

\subsection{Converse for Deterministic Coding}\label{subsection:dfrconverse}


Assume $(R_1,\cdots,R_m)$ is deterministic-fixed-rate achievable. The converse for randomized coding holds equally well here, so $(R_1,\cdots,R_m)\in\Rcal^*_{\text{rfr}}$. We prove by contradiction that $(R_1,\cdots,R_m)\in\Rcal^*_{\text{dfr}}$ as well. Suppose $(R_1,\cdots,R_m)\in\Rcal^*_{\text{rfr}}\backslash\Rcal^*_{\text{dfr}}$, meaning that for some $\Scal_1,\Scal_2\in\Hmsc$, there exists $r\in R(\Scal_2)$ such that $H_r(X_{\Scal_1\cap\Scal_2}|W)=0$ but $R_{\Scal_1\cap\Scal_2}\not\in\text{SW}(X_{\Scal_1\cap\Scal_2})$. Consider the case that $\Hcal=\Scal_1$ and $r$ is such that $H_r(\Scal_1\cap\Hcal|W)=0$. Thus the traitors always have access to $X_{\Scal_1\cap\Hcal}^n$.

For all $\Scal\in\Hmsc$, let $D(X_\Scal)$ be the subset of $T_\eps^n(X_\Scal)$ such that all sequences in $D$ are decoded correctly if $\Scal^c$ are the traitors and no matter what messages they send. Thus the probability that $X_\Scal^n\in D(X_\Scal)$ is large. Let $D(X_{\Scal_1\cap\Hcal})$ be the marginal intersection of $D(X_{\Scal_1})$ and $D(X_\Hcal)$. That is, it is the set of sequences $x_{\Scal_1\cap\Hcal}^n$ such that there exists $x_{\Scal_1\backslash\Hcal}^n$ and $x_{\Hcal\backslash\Scal_1}^n$ with $(x_{\Scal_1\cap\Hcal}^nx_{\Scal_1\backslash\Hcal}^n)\in D(X_{\Scal_1})$ and $(x_{\Scal_1\cap\Hcal}^nx_{\Hcal\backslash\Scal_1})\in D(X_\Hcal)$. Note that with high probability $X_{\Scal_1\cap\Hcal}^n\in D(X_{\Scal_1\cap\Hcal})$. Suppose $X_{\Scal_1\cap\Hcal}^n\in D(X_{\Scal_1\cap\Hcal})$ and $(X_{\Scal_1\cap\Hcal}^nX_{\Hcal\backslash\Scal_1}^n)\in D(X_\Hcal)$, so by the definition of $D$, $\hat{X}_{\Scal_1\cap\Hcal}^n=X_{\Scal_1\cap\Hcal}^n$. Since $R_{\Scal_1\cap\Hcal}\not\in\text{SW}(X_{\Scal_1\cap\Hcal})$, there is some $x'^n_{\Scal_1\cap\Hcal}\in D(X_{\Scal_1\cap\Hcal})$ mapping to the same codewords as $X_{\Scal_1\cap\Hcal}$ such that $x'^n_{\Scal_1\cap\Hcal}\ne X_{\Scal_1\cap\Hcal}^n$. Because the traitors have access to $X_{\Scal_1\cap\Hcal}$, they can construct $x'^n_{\Scal_1\cap\Hcal}$, and also find $x'^n_{\Scal_1\backslash\Hcal}$ such that $(x'^n_{\Scal_1\cap\Hcal}x'^n_{\Scal_1\backslash\Hcal})\in D(X_{\Scal_1})$. If the traitors report $x'^n_{\Scal_1\backslash\Hcal}$, then we have a contradiction, since this situation is identical to that of the traitors being $\Scal_1^c$, in which case, by the definition of $D$, $\hat{X}_{\Scal_1\cap\Hcal}^n=x'^n_{\Scal_1\cap\Hcal}$.

\subsection{Achievability for Deterministic Coding}

Fix $(R_1,\cdots,R_m)\in\Rcal^*_{\text{dfr}}$. Our achievability scheme will be a simple extension of the random binning proof of the Slepian-Wolf theorem given in \cite{Cover:75IT}. Each encoding function $f_i:\Xcal_i^n\to\{1,\cdots,2^{nR_i}\}$ is constructed by means of a random binning procedure. Decoding is then performed as follows.  For each $\Scal\in\Hmsc$, if there is at least one $x_\Scal^n\in T_\eps^n(X_\Scal)$ matching all received codewords from $\Scal$, let $\hat{x}_{i,\Scal}^n$ be one such sequence for all $i\in s$. If there is no such sequence, leave $\hat{x}_{i,\Scal}^n$ null. Note that we produce a separate estimate $\hat{x}_{i,\Scal}^n$ of $X_i^n$ for all $\Scal\ni i$. Let $\hat{x}_i^n$ equal one non-null $\hat{x}_{i,\Scal}^n$.

We now consider the probability of error. With high probability, $\hat{x}_{i,\Hcal}^n=X_i^n$ for honest $i$. Thus all we need to show is that for all other $\Scal\in\Hmsc$ with $i\in\Scal$, $\hat{x}_{i,\Scal}$ is null or also equal to $X_i^n$. Fix $\Scal\in\Hmsc$. If there is some $r\in R(\Scal)$  with $H_r(X_{\Hcal\cap\Scal}|W)=0$, then by the definition of $\Rcal^*_{\text{dfr}}$, $R_{\Hcal\cap\Scal}\in\text{SW}(X_{\Hcal\cap\Scal})$. Thus with high probability the only sequence $x_{\Hcal\cap\Scal}^n\in T_\eps^n(X_{\Hcal\cap\Scal})$ matching all received codewords will be $X_{\Hcal\cap\Scal}^n$, so $\hat{x}_{i,\Scal}^n=X_i^n$ for all $i\in\Hcal\cap\Scal$.

Now consider the case that $H_r(X_{\Hcal\cap\Scal}|W)>0$ for all $r\in R(\Scal)$. For convenience, let $Y=X_{\Hcal\cap\Scal}$ and $Z=X_{\Tcal}$. Let $R_Y=\sum_{i\in\Hcal\cap\Scal}R_i$ and $R_Z=\sum_{i\in\Tcal}R_i$. Since $R_\Scal\in\text{SW}(X_\Scal)$, $R_Y+R_Z\ge H(YZ)+\eta$ for some $\eta$. Let $b_Y(y^n)$ be the set of sequences in $\Ycal^n$ that map to the same codewords as $y^n$, and let $b_Z\subset\Zcal^n$ be the set of sequences mapping to the codewords sent by the traitors. Then $Y$ may be decoded incorrectly only if there is some $y'^n\in b_Y(Y^n)$  and some $z^n\in b_Z$ such that $y'^n\ne Y^n$ and $(y'^nz^n)\in T_\eps^n(YZ)$. For some $w^n\in \Wcal^n$,
\begin{align}
&\Pr(\exists y'^n\in b_Y(Y^n)\backslash\{Y^n\},z^n\in b_Z:\nonumber\\
&\qquad\qquad\qquad\qquad\qquad\qquad(y'^nz^n)\in T_\eps^n(YZ)|W^n=w^n)\nonumber\\
&\le\Pr(Y^n\not\in T_\eps^n(Y|w^n)|W^n=w^n)+\sum_{y^n\in T_\eps^n(Y|w^n)}p(y^n|w^n)\displaybreak\nonumber\\
&\quad\cdot \textbf{1}\{\exists y'^n\in b_Y(y^n)\backslash\{y^n\},z^n\in b_Z:(y'^nz^n)\in T_\eps^n(YZ)\}\nonumber\\
&\le\eps+2^{-n(H(Y|W)-\eps)}\sum_{z^n\in b_Z\cap T_\eps^n(Z)}k_3(z^n,w^n)\label{eq:k3}
\end{align}
where
\begin{multline*}
k_3(z^n,w^n)\triangleq|\{y^n\in T_\eps^n(Y|w^n):\\
\exists y'^n\in b_Y(y^n)\cap T_\eps^n(Y|z^n)\backslash\{y^n\}\}|.\end{multline*}
On average, the number of typical $y^n$ put into a bin is at most $2^{n(H(Y)-R_Y+\eps)}$, so we can use \eqref{eq:binomialbound} to assume with high probability than no more than $2^{n(H(Y)-R_Y+2\eps)}$ are put into any bin. Note that
\begin{multline*}\sum_{z^n\in T_\eps^n(Z)}k_3(z^n,w^n)\\
\begin{aligned}&\le \sum_{z^n\in T_\eps^n(Z)}\sum_{y^n\in T_\eps^n(Y|w^n)}|b_Y(y^n)\cap T_\eps^n(Y|z^n)\backslash\{y^n\}|\\
&=\sum_{y^n\in T_\eps^n(Y|w^n)}\sum_{y'^n\in b_Y(y^n)\cap T_\eps^n(Y|z^n)\backslash\{y^n\}}|T_\eps^n(Z|y'^n)|\\
&\le 2^{n(H(Y|W)+\eps)}2^{n(H(Y)-R_Y+2\eps)}2^{n(H(Z|Y)+\eps)}\\
&=2^{n(H(YZ)+H(Y|W)-R_Y+4\eps)}.\end{aligned}\end{multline*}
The average $k_3$ sum over typical $z^n$ in a given bin is thus 
\[2^{n(H(YZ)+H(Y|W)-R_Y-R_Z+4\eps)}\le 2^{n(H(Y|W)+4\eps-\eta)}.\]
We can use an argument similar to that in Section~\ref{subsection:imperfect}, partitioning $T_\eps^n(Z)$ into different $l$ values, to show that with high probability, since $H(Y|W)>0$, for all bins $b_Z$,
\[\sum_{z^n\in T_\eps^n(Z)\cap b_Z}k_3(z^n,w^n)\le 2^{n(H(Y|W)+5\eps-\eta)}.\]
Applying this to \eqref{eq:k3} gives
\begin{multline*}
\Pr(\exists y'^n\in b_Y(Y^n)\backslash\{y^n\},z^n\in b_Z:\\
(y'^nz^n)\in T_\eps^n(YZ)|W^n=w^n)\le \eps+2^{n(6\eps-\eta)}.\end{multline*}
Letting $\eta>6\eps$ ensures that the probability of error is always small no matter what bin $b_Z$ the traitors choose.

\subsection{Achievability for Randomized Coding}

We perform essentially the same coding procedure as with deterministic coding, expect we also apply randomness in a similar fashion as with variable-rate coding. The only difference from the deterministic coding scheme is that each sensor has a set of $C$ identically created subcodebooks, from which it randomly chooses one, then sends the chosen subcodebook index along with the codeword. Decoding is the same as for deterministic coding. An argument similar to that in \hbox{Section~\ref{subsection:errorprob}} can be used to show small probability of error.

\section{Conclusion}\label{section:conclusion}

We gave an explicit characterization of the region of achievable rates for a Byzantine attack on distributed source coding with variable-rate codes, deterministic fixed-rate codes, and randomized fixed-rate codes. We saw that a different set of rates were achievable for the three cases, and gave converse proofs and rate achieving coding schemes for each. Variable-rate achievability was shown using an algorithm in which sensors use randomness to make it unlikely that the traitors can fool the coding process.

Much more work could be done in the area of Byzantine network source coding. Multiterminal rate distortion \cite{Tung:thesis,Berger:Chapter78} could be studied, or other topologies, such as side information. However, perhaps the biggest drawback in this paper is that, as we discussed in the introduction, because the traitors cannot in general be identified, it is difficult to imagine applications that do not require some post processing of the source estimates, for example to estimate some underlying process. Thus it would make sense to solve the coding and estimation problems simultaneously, such as in the the CEO problem \cite{Berger&etal:96IT}.

\bibliographystyle{ieeetr}
{
\bibliographystyle{ieeetr}
\bibliography{\ACSP/Reference/Bibs/Journal,\ACSP/Reference/Bibs/Conf,%
\ACSP/Reference/Bibs/Book,\ACSP/Reference/Bibs/Misc,\ACSP/Reference/Bibs/ACSP-J,\ACSP/Reference/Bibs/ACSP-C}
}

\end{document}